\documentclass[12pt,a4paper]{article}
\usepackage{amsmath,amsfonts,amssymb,mathtools}
\usepackage{array}
\usepackage[caption=false,font=normalsize,labelfont=sf,textfont=sf]{subfig}
\usepackage[colorlinks=true]{hyperref}
\usepackage{relsize}
\usepackage[linesnumbered,ruled,vlined]{algorithm2e}
\usepackage{makecell}
\usepackage{amsthm}
\usepackage{textcomp}
\usepackage{stfloats}
\usepackage{url}
\usepackage{xcolor}
\usepackage{verbatim}
\usepackage{graphicx}
\usepackage{scalerel}
\usepackage{romannum}
\usepackage{bm}
\usepackage[left=2.00cm, right=2.00cm, top=2.00cm, bottom=2.00cm]{geometry}

\SetCommentSty{mycommfont}
\newcommand\undermat[2]{%
	\makebox[0pt][l]{$\smash{\underbrace{\phantom{%
					\begin{matrix}#2\end{matrix}}}_{\text{$#1$}}}$}#2}
\SetKwInput{KwInput}{Input}                
\SetKwInput{KwOutput}{Output}              

\def\BibTeX{{\rm B\kern-.05em{\sc i\kern-.025em b}\kern-.08em
		T\kern-.1667em\lower.7ex\hbox{E}\kern-.125emX}}
\usepackage{balance}
\usepackage{lipsum}
\usepackage[usestackEOL]{stackengine}
\usepackage[nottoc,notlot,notlof]{tocbibind}
\usepackage[utf8]{inputenc}
\usepackage{orcidlink}
\newtheorem{theorem}{Theorem}
\newtheorem{corollary}{Corollary}[theorem]
\newtheorem{lemma}[theorem]{Lemma}
\newtheorem{defn}{Definition}

\newtheorem{exmp}{Example}
\newtheorem{proposition}{Proposition}
\newtheorem*{remark}{Remark}
\date{}

\begin{document}
	\title{Additive one-rank hull codes over finite fields}
\author{Astha Agrawal\orcidlink{0000-0002-4583-5613} and R. K. Sharma\orcidlink{0000-0001-5666-4103}}
\markboth{}
{Additive One-rank Hull Codes over Finite Fields}
\maketitle
\vspace{-12mm}
	\begin{center}
	\noindent {\small Department of Mathematics, Indian Institute of Technology Delhi, New Delhi 110016, India.}
	\end{center}
	\footnotetext[1]{{\em E-mail addresses:} \url{asthaagrawaliitd@gmail.com} (A. Agrawal), \url{rksharmaiitd@gmail.com} (R. K. Sharma).}
	\begin{abstract}
	This article explores additive codes with one-rank hull, offering key insights and constructions.   The article  introduces a novel approach to finding one-rank hull codes over finite fields by establishing a connection between self-orthogonal elements and solutions of quadratic forms. It also provides a precise count  of self-orthogonal elements for any duality over the finite field $\mathbb{F}_q$, particularly  odd primes. Additionally, construction methods for small rank hull codes are introduced. The highest possible  minimum distance among additive one-rank hull codes is denoted by $d_1[n,k]_{p^e,M}$.   The value of $d_1[n,k]_{p^e,M}$ for $k=1,2$ and $n\geq 2$ with respect to any duality $M$ over any finite field $\mathbb{F}_{p^e}$ is determined. Furthermore, the new quaternary one-rank hull codes are identified over non-symmetric dualities with better parameters than symmetric ones.  
		\end{abstract}
	\section{Introduction}
	In 1990, Assmus and Key \cite{ASSMUSe_KEY} introduced the concept of the Euclidean hull to study finite projective planes. This concept helps to classify these planes by looking at codes generated from the characteristic functions of their lines. The Euclidean hull of a linear code, denoted as $Hull_E(C)$, is the intersection of the code with its dual, i.e., $Hull_E(C)=C\cap C^{\perp}$, where $C^{\perp}$ represents the Euclidean dual of the code 	$C$.  The hull of a code is essential for understanding the automorphism group of a linear code. Also, it is useful in evaluating the complexity of certain algorithms for checking the permutation equivalence of two linear codes. This is particularly true when the hull's dimension is small. These findings have been widely acknowledged in the literature \cite{ compute_auto, Permutation_group, permut_equiv, Auto_linear}. If $Hull_E(C)=\{\textbf{0}\}$, then $C$ is known as linear complementary dual code (LCD)  and which has the smallest hull dimension. Following LCD codes in terms of hull dimension are one-dimensional codes with the second smallest hull dimension among linear codes. Numerous research papers have focused on linear codes characterized by small hulls, including LCD  and one-dimensional hull codes, see \cite{nchainlcd,  construct_one, optimal, LCDMASSEY, One_gauss_sum, Galois}. In this work, we are looking at codes over additive abelian groups. We make it more flexible by considering many dualities.  As a result, we generalize the permitted inner products as well as the ambient space. 
	
	Additive codes are the codes defined over additive groups.  Additive codes over finite fields are defined by using the additive group of a finite field.  All linear codes over finite fields are additive codes, but the converse is not true. Therefore,  the class of additive codes is bigger than the class of linear codes.	Additive codes over finite fields have attracted significant attention because of their relevance to quantum error correction and quantum computing \cite{ quantum_nonbinary, quantum_additiveskew, quantum_GF(4),  quantum_constacyclic}. Additive codes allow the range of dualities defined by the character of a finite abelian group \cite{Orthogonality}. While the usual dualities (Euclidean, Hermitian, Trace) serve as one example of such dualities, it is important to note that  various other potential forms of dualities  can be used in this context. In essence, additive codes can employ different types of dualities, not limited to the standard ones, to achieve their specific coding objectives. Delsarte  \cite{Delsarte}  introduced additive codes in 1973 using association schemes. For symmetric and non-symmetric dualities, an additive code satisfies the cardinality condition. This means that the product of the number of elements in the additive code and its dual, with respect to any duality, must be equal to the total number of elements in the space \cite{dualities_abeliangroup}. The idea of finding the optimal parameters for additive quaternary codes was initially proposed by Blokhuis and Brouwer \cite{small_additive}. Later, in \cite{short_additive}, Bierbrauer et al. successfully identified the optimal parameters for additive quaternary codes for cases where $n$ is less than or equal to 13. 
	
	Recently, Dougherty et al. \cite{ACDgroupcharacter,onerank} generalized the idea of LCD and one dimensional hull codes to additive complementary dual codes (ACD) and one-rank hull codes, respectively. They looked at quaternary ACD and one-rank hull codes in their study. They determined optimal parameters for  these codes over symmetric dualities for length $n\leq 10$. Subsequently, Agrawal and Sharma \cite{ACDovernon-symmetric}   introduce a specific subset of non-symmetric dualities referred to as skew-symmetric dualities over finite fields, and study ACD codes on these dualities.  Since small hull codes have their importance in evaluating the complexity of certain algorithms. Inspired by this fact, we study one rank hull codes with respect to arbitrary dualities over additive group of finite fields. 

 This article includes several key sections. In section \ref{Preliminaries}, we recall the basic definitions and notations required for this article. In Section \ref{Characterization}, we give  the characterization of $Hull(C)$ using the generator matrix of $C$. Moving on to Section \ref{Existence}, we show a relation between self-orthogonal elements and solutions of quadratic form $\mathcal{Q}=0$ associated with the duality matrix.  We establish the 
  existence of a one-rank hull code in relation to any duality over finite fields. We provide a precise count of self-orthogonal elements for any duality over the finite field $\mathbb{F}_q$, particularly for odd primes. Section \ref{Constructions} introduces some construction methods for ACD  and one-rank hull codes. In Section \ref{Bounds}, we study the highest minimum distance $d_1[n,k]_{p^e,M}$ among all additive $[n, p^k]$ codes with one-rank hull. We determine the value of $d_1[n,k]_{p^e,M}$  for $k=1,2$ and $n\geq 2$ with respect to any duality $M$ over any finite field $\mathbb{F}_{p^e}$.   We find some better parameters for quaternary codes over non-symmetric dualities than symmetric ones. Finally, the article concludes by summarizing the findings and outlining potential future directions for research.
	\section{Preliminaries}\label{Preliminaries}
We  briefly review fundamental definitions and notations necessary for the upcoming sections of this article.
\subsection{The Finite Field and Codes over Finite Fields}
	Let $\mathbb{F}_q$ be a finite field of order $q=p^e$, where $p$ is a prime and $e$ is a positive integer. It is worth noting that additive group $(\mathbb{F}_{p^e},+)$ is finitely generated abelian group and it can be expressed as $\mathbb{F}_{p^e}=\left\langle x_1,x_2,\dots,x_e \right\rangle $, where $x_1, x_2, \dots , x_e$ are generators of order $p$. Let $\mathbb{F}^*_q=\mathbb{F}_q\setminus\{0\}$ represents the multiplicative group of finite field $\mathbb{F}_q$. If we have a positive integer $n$, then $\mathbb{F}^n_q$ refers to an $n$ dimensional vector space over $\mathbb{F}_q$ and $(\mathbb{F}^n_q,+)$ is an abelian group. The Hamming weight of a vector $\textbf{x}$ is determined by counting the number of its non-zero components, denoted as $wt(\textbf{x})$.  The Hamming weight of the difference between two vectors $\textbf{x}$ and $\textbf{y}$ is used to determine the distance between two vectors as $d(\textbf{x}, \textbf{y}) = wt(\textbf{x} - \textbf{y})$. A code of length $n$ over $\mathbb{F}_q$ is a non-empty subset of $\mathbb{F}_q^n$. A vector of a code is called a codeword. The minimum distance of a code $C$ is determined by finding the smallest possible value for the distance between any two distinct codewords within $C$, represented as $d=\min\{d(\textbf{x}, \textbf{y}) | \textbf{x}, \textbf{y} \in C\  \text{and} \ \textbf{x} \neq \textbf{y}\}$.
	\begin{defn}
		Let G be a finite abelian group.  An additive code $C$ over $G$ of length $n$  is an additive subgroup of $G^n$.
	\end{defn}
	\begin{defn}
		An $[n,k,d]$ linear code over $\mathbb{F}_q$ is defined as a $k$-dimensional subspace of the vector space $\mathbb{F}^n_q$  with distance $d$. 
	\end{defn}
	It is important to note that additive codes over $\mathbb{F}_q$ are not necessarily subspaces; instead, they are subgroups of the additive group $(\mathbb{F}^n_{p^e},+)$. When we say a code is additive over a field $\mathbb{F}_{p^e}$, it implies that the code is considered linear over $\mathbb{F}_p$. For example, quaternary additive codes over $\mathbb{F}_4 =\{0,1,\omega,1+\omega\}$ are understood to be $\mathbb{F}_2$-linear code. Any linear code over $\mathbb{F}_q$ is additive code, but the converse is not true. For example, the code $\{(0,0), (1,0),(\omega,0),(1+\omega,0)\}$ is a linear code as well as additive code over $\mathbb{F}_4$, and the code $\{(0,0),(1,0)\}$ is an additive code  but not a linear code over $\mathbb{F}_4$. An additive code $C$ over $\mathbb{F}_{p^e}$ of length $n$ is a $\mathbb{F}_p$-linear code with size $p^k$ for some $0\leq k\leq ne$ and it is denoted as $[n,p^k]$ code. If minimum distance of $C$ is $d$, then $C$ is called an $[n,p^k,d]$ additive code.
	\begin{defn}
		A generator matrix $\mathcal{G}$ of an $[n,p^k]$ additive code is a $k\times n$ matrix over $\mathbb{F}_{p^e}$ such that every codeword of $C$ is $\mathbb{F}_p$-linear combination of the rows of $\mathcal{G}$, i.e., $C=\{\textbf{u} \mathcal{G}| \textbf{u}\in \mathbb{F}^k_p\}$. Note that $p$-rank of the matrix $\mathcal{G}$ is $k$, where $p$-rank means the number of   linearly independent rows over $\mathbb{F}_p$.
		\begin{defn}
		 The $p$-rank of the additive code $C$, denoted by $rank_p(C)$, is equal to $p$-rank of its generator matrix $\mathcal{G}$. In this context,  $p$ is typically understood; it is often omitted from the notation for simplicity.
		\end{defn}
	\end{defn}
	\begin{exmp}
		Let $\mathbb{F}_{4}=\{0,1,\omega,\omega+1\}$. Suppose $\mathcal{G}=\begin{pmatrix}
			1&\omega&1&1+\omega&\omega\\
			\omega&1+\omega&\omega&1&1+\omega
		\end{pmatrix}$ is a generator matrix of an additive code $C$ over $\mathbb{F}_4$. Then $C=\{(0,0,0,0,0),(1,\omega,1,1+\omega,\omega), (	\omega,1+\omega,\omega,1,1+\omega),(1+\omega,1,1+\omega,\omega,1)\}$ is a $[5,2^2,5]$ additive code over $\mathbb{F}_4$. We observe that 2-rank of the matrix  $\mathcal{G}$ is $2$ but 4-rank of $\mathcal{G}$ is 1. 
	\end{exmp}
	\subsection{Duality}
	Let's now discuss the concept of duality as applied to additive codes. It is important to note that there exist various  potential forms of duality that can be used in this context. These dualities are defined by the character of a finite abelian group.
	\begin{defn}
		A character of a group $G$ is a homomorphism from $G$ to $\mathbb{C}^*$, the multiplicative group of complex numbers.
	\end{defn}
	Although some sources describe the character as a homomorphism from $G$ to $\mathbb{Q}/ \mathbb{Z}$, we will adopt the complex variant, as was previously employed in \cite{ACDovernon-symmetric}, \cite{Algebraiccoding}, \cite{ACDgroupcharacter},  and \cite{Orthogonality}. Let $\widehat{G}$ be the collection of all characters of $G$, denoted as $\widehat{G} = \{\chi| \chi \ \text{is a character of} \ G\} $. We observe that $\widehat{G}$ forms a group, and this group is isomorphic to $G$. However, it is worth noting that the isomorphism between $\widehat{G}$ and $G$ is not uniquely determined. In fact, the specific isomorphism, we select between $\widehat{G}$ and $G $ is what allows us to establish the concept of duality. 
	
	By establishing an isomorphism between $G$ and its character group $\widehat{G}$, we proceed to construct a character table or duality $M$ as follows. Let $x_1,x_2,\dots,x_e$ be the elements of the group $G$, and $\phi : G \to \widehat{G}$ denotes the isomorphism. We define the image of $x_i$ under this isomorphism as $\chi_{x_i}$, which is essentially $\phi(x_i)$. The character table 
	$M$ is structured such that its rows are labeled by $\chi_{x_i}$, while its columns are labeled by elements of the group. Indeed, the entry in the $ij$-th position of the character table $M$ is given by $\chi_{x_i}(x_j)$.  By utilizing these dualities, we can establish the concept of orthogonality in additive codes.
	\begin{defn}
		The orthogonal code of an additive code $C$ over $G$ with respect to duality $M$ is defined as $C^M=\{(u_1,u_2,\dots,u_n)\in G^n|\prod_{i=1}^{n}\chi_{u_i}(v_i)=1 \ \forall (v_1,v_2,\dots,v_n)\in C\}$, where $ \chi_{u_i}$ represents the character associated to the element $u_i$ under the duality $M$.
	\end{defn}
	It is  correct to observe that $C^M$ is always an additive code, regardless of $C$. Additive codes consistently satisfy the cardinality condition for all types of dualities, i.e.,  $|C||C^M|=|G^n|$ \cite{Algebraiccoding}. It is  crucial to recognize that the choice of duality within a group significantly impacts the definition of the orthogonal code.  In fact, a code can be identical to its dual in one duality but not in another.
	If $M$ is a duality then $M^T$ (transpose of the character table) is also a duality, indicating the existence of an isomorphism $\phi'$ from $G$ to $\widehat{G}$ that corresponds to $M^T$.
	\begin{defn}
		Let $G$ be a finite abelian group and $M$  be a duality over $G$. A duality $M$ is said to be a symmetric duality if $\chi_x(y)=\chi_y(x)$, for all $x,y\in G$. This implies $M=M^T$. A duality $M$ is said to be a skew-symmetric duality if $\chi_x(x)=1$, for all $x\in G$. For a skew-symmetric duality we have $\chi_x(y)=(\chi_y(x))^{-1}$, for all $x,y\in G$. For further details on skew-symmetric dualities, see \cite{ACDovernon-symmetric}.
	\end{defn}
	If $C$ is an additive code  over $G$, then it is worth noting that for both symmetric and skew-symmetric dualities, we have $(C^M)^M=C$ (see \cite[Theorem 7]{ACDovernon-symmetric}). Some dualities are neither symmetric nor skew-symmetric. A character in the context of $G=(\mathbb{F}_{p^e},+)$ is a mapping $\chi: \mathbb{F}_{p^e}\to P\subset \mathbb{C}^*$, where $\chi(x)$ yields a value $\xi$ from the set of all $p$-th root of unity, denoted as $P$. This set $P$ forms a cyclic group of order $p$.
	
	Consider vectors, $\textbf{a}=(a_1,a_2,a_3,\dots,a_n),\ \textbf{b}=(b_1,b_2,b_3,\dots,b_n)\in G^n$. Then $\chi_{\textbf{a}}(\textbf{b})$ is defined  as  the product of individual character values, which can be expressed as $\chi_{\textbf{a}}(\textbf{b})=\prod_{i=1}^{n}\chi_{a_i}(b_i)$.  We say  $\textbf{a}$ is orthogonal to  $\textbf{b}$ with respect to duality $M$ if and only if $\chi_{\textbf{a}}(\textbf{b})=1$. A vector $\textbf{a}$ is said to be self-orthogonal if $\chi_{\textbf{a}}(\textbf{a})=1$.
	\begin{exmp}
		Over $\mathbb{F}_{3^2}$, there are 48 dualities, of which 18 are symmetric and two are skew-symmetric; the rest are neither symmetric nor skew-symmetric. Let's see one example of each. Let $(\mathbb{F}_{3^2},+)=\left\langle 1,\nu \right\rangle=\{0,1,2,\nu,2\nu,1+\nu,2+\nu,1+2\nu,2+2\nu\}$.
		\begin{enumerate}
			\item  $M_1$ is defined as, $\chi_1(1)=\xi$, $\chi_1(\nu)=1$, $\chi_{\nu}(1)=1$ and $\chi_{\nu}(\nu)=\xi^2$.
			\item  $M_2$ is defined as, $\chi_1(1)=1$, $\chi_1(\nu)=\xi$, $\chi_{\nu}(1)=\xi^2$ and $\chi_{\nu}(\nu)=1$.
			\item $M_3$ is defined as, $\chi_1(1)=1$, $\chi_1(\nu)=\xi$, $\chi_{\nu}(1)=\xi^2$ and $\chi_{\nu}(\nu)=\xi$.
		\end{enumerate}
		It is clear that $M_1$ is symmetric, $M_2$ is skew-symmetric, and $M_3$ is neither symmetric nor skew-symmetric duality. We will refer to these dualities in subsequent sections to understand some examples. 
	\end{exmp}
	\begin{defn}
		For an additive code $C$, if $C\subseteq C^M$, then we say $C$ is a self-orthogonal additive code with respect to duality $M$. And if $C=C^M$, then we say $C$ is  a self-dual code with respect to duality $M$. 
	\end{defn}
	It is  important to emphasize that an additive code $C$  can exhibit self-orthogonality or self-duality with respect to a particular duality $M$, but this property doesn't necessarily hold for  other duality $M'$. Each duality structure may lead to different notions of self-orthogonality and self-duality.
	\begin{defn}
		The hull of an additive code $C$ over a finite group $G$, considering any duality $M$, is the intersection of the code $C$ with its dual under M, i.e., $Hull_M(C)=C\cap C^M$. The duality $M$ is clear from the context; hence, omit it from the notation.
\end{defn}	
When $C$ is an additive code, $Hull(C)$ is also an additive code. An additive code $C$ is said to be an additive complementary dual (ACD) code if $Hull(C)=\{\textbf{0}\}$.  An additive  code $C$ is said to be a one-rank hull code if $rank(Hull(C))=1$.


 \section{Characterization of one-rank hull codes}\label{Characterization}
 Let $\mathcal{G}$ be a $k\times n$  matrix then the matrix product $\mathcal{G}\odot_M \mathcal{G}^T$ is defined as, $$(\mathcal{G}\odot_M \mathcal{G}^T)_{ij}=\chi_{\mathcal{G}_i}(\mathcal{G}_j)=\xi^{k_{ij}},$$ where $\mathcal{G}_i$ and $\mathcal{G}_j$ are $i$-th and $j$-th row of the matrix $\mathcal{G}$ and $\chi_{\mathcal{G}_i}$ represents the character associated to vector $\mathcal{G}_i$ under duality $M$. Then the matrix $\log_\xi(\mathcal{G}\odot_M \mathcal{G}^T)$ is defined as
 $$(\log_\xi(\mathcal{G}\odot_M \mathcal{G}^T))_{ij}=\log_\xi(\mathcal{G}\odot_M \mathcal{G}^T)_{ij}=\log_{\xi}(\xi^{k_{ij}})\equiv k_{ij}\mod p,$$ where $\log_{\xi}$ represents the principal branch of complex logarithm to the base $\xi$ and $\xi$ is the primitive $p$-th root of unity with $0 \leq k_{ij} \leq p-1$. Here, we shall give a characterization of the  hull of an additive code $C$ in terms of its generator matrix.  The following result  has previously been mentioned in \cite{onerank}; however, we are now presenting a detailed proof of this result in our setting. As, this proof holds significance in establishing the Theorem \ref{skew-symmetric k odd}. 
	\begin{theorem}\label{rankhull}
		Let $\mathcal{G}$ be a $k\times n$  generator matrix of an additive code $C$ over $\mathbb{F}_q$. Then $rank(\log_\xi(\mathcal{G}\odot_M \mathcal{G}^T))=k-rank(Hull(C))$.
		\end{theorem}
	\begin{proof}
 	Let $\mathcal{G}_i$ be the $i$-th row of the matrix $\mathcal{G}$. Let $\textbf{c}\in Hull(C)$, then $\textbf{c}=\sum_{i=1}^{k}n_i\mathcal{G}_i$ for some $n_i\in \mathbb{F}_p$ and $ \chi_{\sum_{i=1}^{k}n_i\mathcal{G}_i}(\mathcal{G}_j)=1$ for all $j=1,2,\dots, k$. This implies that $ \prod_{i=1}^{k}\chi_{\mathcal{G}_i}(\mathcal{G}_j)^{n_i}=1$ for all $j=1,2,\dots, k$. Let $\chi_{\mathcal{G}_i}(\mathcal{G}_j)=\xi^{r_{ij}}$ for all $1\leq i,j\leq k$, where $ \xi $ is the primitive $ p $-th root of unity and $r_{ij}\in\mathbb{F}_p$. Now, we have $\xi^{\sum_{i=1}^{k}r_{ij}n_i}=1$ for all $j=1,2,\dots, k$. It gives that $\sum_{i=1}^{k}r_{ij}n_i\equiv 0 \mod p $ for all $j=1,2,\dots, k$. Hence, we have the following system of equations
 	$$\begin{pmatrix}
 	r_{11}&r_{21}&\dots&r_{k1}\\
 	r_{12}&r_{22}&\dots&r_{k2}\\
 
 	\vdots&\vdots&\ddots&\vdots\\
 	r_{1k}&r_{2k}&\dots&r_{kk}
 \end{pmatrix}_{k\times k}  \begin{pmatrix}
 	n_1\\n_2\\ \vdots\\n_{k}
 \end{pmatrix}\equiv\begin{pmatrix}
 	0\\0\\ \vdots\\0
 \end{pmatrix}\mod p.$$
In particular, we have,
$$(\log_\xi(\mathcal{G}\odot_M \mathcal{G}^T))^T\begin{pmatrix}
	n_1\\n_2\\ \vdots\\n_{k}
\end{pmatrix}\equiv\begin{pmatrix}
	0\\0\\ \vdots\\0
\end{pmatrix} \mod p.$$
If rank of the matrix $\log_\xi(\mathcal{G}\odot_M \mathcal{G}^T)$ is $r$ then we have $k-r$ independent solutions of the above system of equations over $\mathbb{F}_p$. This implies that there are $k-r$ independent vectors in $Hull(C)$, i.e., $rank(Hull(C))=k-r$.  Hence, the result follows.
\end{proof}
The following corollary directly follows from the theorem.
\begin{corollary}\label{one-rank hull}
	Let $C$ be an additive code over $\mathbb{F}_q$ with generator matrix $\mathcal{G}$ then $C$ is a one-rank hull code if and only if $rank(\log_\xi(\mathcal{G}\odot_M \mathcal{G}^T))=k-1$.
	\end{corollary}
	The following theorem gives an interesting non-existence result of additive one-rank hull code with respect to skew-symmetric dualities
	in terms of the number of generators.
	\begin{theorem}\label{skew-symmetric k odd}
			Let $\mathcal{G}$ be an $k\times n$  generator matrix of an additive code $C$ over $\mathbb{F}_q$. If $C$ is a one-rank hull code under skew-symmetric duality $M$, then $k$ is odd.
	\end{theorem}
\begin{proof}
	If $C$ is a one-rank hull code under skew-symmetric duality $M$, then by the  Theorem \ref{rankhull} the matrix $\log_\xi(\mathcal{G}\odot_M \mathcal{G}^T)$ is a skew-symmetric matrix of  rank $ k-1 $. We know that the rank of the skew-symmetric matrix is always even; hence, $k$ must be odd.
\end{proof}
\begin{theorem}
	If $C$ is a one-rank hull code of length  $n$  and $D$ is  an ACD code of  length $m$,  then $C\times D$ is a  one-rank hull code of length $n+m$.
\end{theorem}
\begin{proof}
	Let   $Hull(C)=\langle \textbf{x} \rangle$, for some $\textbf{x}\in \mathbb{F}^n_q$ and $Hull(D)=\{\textbf{0}\}$ then $Hull(C\times D)=\langle \textbf{x}\rangle\times \{\textbf{0}\}=\langle (\textbf{x},\textbf{0})\rangle$. Hence, $C\times D$ is a one-rank hull code of length $n+m$.
\end{proof}
\begin{lemma}\label{C is one-rank wrt M^T}
	An additive code $C$ is a one-rank hull code with respect to duality $M$  if and only if it is a one-rank hull code with respect to duality $M^T$.
\end{lemma}
\begin{proof}
	Let $\mathcal{G}$ be a generator matrix of a one-rank hull code $C$. Suppose $\chi$ and $\chi'$ represents the dualities $M$ and $M^T$, respectively, then $\chi'_a(b)=\chi_b(a)$. From this fact, we have $\mathcal{G}\odot_{M^T} \mathcal{G}^T=(\mathcal{G}\odot_M \mathcal{G}^T)^T$. It follows that, $rank(\log_\xi(\mathcal{G}\odot_{M^T} \mathcal{G}^T))=rank(\log_\xi(\mathcal{G}\odot_M \mathcal{G}^T))^T=rank(\log_\xi(\mathcal{G}\odot_M \mathcal{G}^T))$. Thus, the result follows from Corollary \ref{one-rank hull}.
\end{proof}
\begin{proposition}\label{C^{M^T} is one-rank}
If $C$ is a one-rank hull code with respect to duality $M$, then $C^M$  and $C^{M^T}$ are one-rank hull codes with respect to duality $M^T$ and $M$, respectively.
\end{proposition}
\begin{proof}
	Let $C$ be a one-rank hull code with respect to duality $M$, i.e., $C\cap C^M=\left\langle \textbf{x}\right\rangle$, for some $\textbf{x}\in \mathbb{F}_{p^e}^n$. By above lemma, we note that $C\cap C^{M^T}=\left\langle \textbf{y}\right\rangle $, for some $\textbf{y}\in \mathbb{F}^n_{p^e}$. Then, we have $C^M\cap (C^M)^{M^T}=C^M\cap C=\left\langle \textbf{x}\right\rangle $ and  $C^{M^T}\cap (C^{M^T})^M=C^{M^T}\cap C=\left\langle \textbf{y}\right\rangle $. This completes the proof.
\end{proof}
\begin{remark}
	For symmetric and skew-symmetric dualities, we have $(C^M)^M=C$,  see \cite{ACDovernon-symmetric}. From this fact, we conclude that if $C$ is a one-rank hull code with respect to duality $M$, then $C^M$ and $C^{M^T}$ both are one-rank hull codes under the same duality $M$.
\end{remark}
\section{Existence of one-rank hull codes}\label{Existence}
Suppose that generators of additive group of finite fields $\mathbb{F}_{p^e}$ are $ x_1, x_2,\dots , x_e $ such that $\mathbb{F}_{p^e}=\left\langle x_1, x_2,\dots , x_e\right\rangle$. Let $M$ be a duality over $\mathbb{F}_{p^e}$ with $\chi_{x_i}(x_j)=\xi^{k_{ij}}$ for all $1\leq i,j\leq e$, where $\xi$ is the primitive $ p $-th root of unity and $0\leq k_{ij}\leq p-1$. Then we can define  $e\times e$ matrix, say Duality matrix, $D=[k_{ij}]$ over $\mathbb{F}_p$ for any duality $M$. For symmetric dualities, the duality matrix $D$ is symmetric. 

  Define a quadratic form $\mathcal{Q}: \mathbb{F}_p^e\to \mathbb{F}_p$ for any   duality $M$ such that $\mathcal{Q}(\textbf{a})=\textbf{a}D\textbf{a}^T$. If $p\neq 2$ then quadratic form $\mathcal{Q}$ can be defined as  $\mathcal{Q}(\textbf{a})=\textbf{a}D\textbf{a}^T=\textbf{a}D'\textbf{a}^T$, where $D'=\left[\frac{k_{ij}+k_{ji}}{2} \right]$ is a symmetric matrix. The quadratic form is said to be non-degenerate if the associated symmetric matrix $D'$ is non-singular and, in particular, $\det \mathcal{Q}=\det D'$.

The following lemma describes an interesting connection between self-orthogonal elements and solutions of quadratic form $\mathcal{Q}=0$.
\begin{lemma}\label{selforthogonal}
	Let $M$ be any duality over $\mathbb{F}_{p^e}$ with duality matrix $D$. Then, for any $x=\sum_{i=1}^{e}a_ix_i$,  $\chi_{x}(x)=1$ if and only if $\mathcal{Q}(\textbf{a})=0$, where $\textbf{a}=(a_1,a_2,\dots,a_e)$.
\end{lemma}
\begin{proof}
	We have  $\chi_{x}(x)=\chi_{\sum_{i=1}^{e}a_ix_i}(\sum_{i=1}^{e}a_ix_i)=\prod_{i,j=1}^{e}\chi_{x_i}(x_j)^{a_ia_j}=\xi^{\sum_{i,j=1}^{e}k_{ij}a_ia_j}=1$ if and only if  $\sum_{i,j=1}^{e}k_{ij}a_ia_j\equiv0\mod p$.	Let $\mathcal{Q}$ be a quadratic form associated with the duality matrix $D$ then $\mathcal{Q}(\textbf{a})=\textbf{a}D\textbf{a}^T=\sum_{i,j=1}^{e}k_{ij}a_ia_j$. Hence, the result follows.
\end{proof}
The following theorem shows the existence of self-orthogonal elements for any duality over $\mathbb{F}_{p^e},\ (e\geq3)$.
\begin{theorem}\label{existence of 1 rank for geq 3}
	Let $\mathbb{F}_{p^e},\ (e\geq3)$, be a finite field, then there  exists one-rank hull code of  length 1 with respect to any duality over $ \mathbb{F}_{p^e}$.
\end{theorem}
\begin{proof}
Let $M$ be a duality over $\mathbb{F}_{p^e},\ (e\geq3)$. Define a quadratic form, say $\mathcal{Q}$, corresponding to the  duality $M$. Then $ \mathcal{Q} $ is a homogeneous polynomial of degree 2 with $e\ (\geq 3)$ variables. By Chevalley's theorem (see \cite[Corollary 6.6]{finitefields}), $ \mathcal{Q} $ has non-trivial solutions over $\mathbb{F}_p$. Therefore, using Lemma \ref{selforthogonal}, we have an element, say $ x \in\mathbb{F}^*_{p^e}$, such that $\chi_{x}(x)=1$. Let $C=\left\langle x \right\rangle $	 be an additive code then $Hull(C)=C$. Hence, $C$ is a one-rank hull code of length 1 for the   duality $ M $ over $\mathbb{F}_{p^e},\ ( e\geq3\ )$. This completes the proof.
\end{proof}
\begin{corollary}
	Over $\mathbb{F}_{p^e},\ (e\geq3)$, there exists one-rank hull code of length $n$ with respect to any duality over $ \mathbb{F}_{p^e}$.
\end{corollary}
\begin{proof}
From Theorem \ref{existence of 1 rank for geq 3}, we have $x \in\mathbb{F}^*_{p^e}$, such that $\chi_{x}(x)=1$. Let $C=\left\langle (x,x,\dots,x)\right\rangle $ be an additive code of length $n$ then $Hull(C)=C$.  Hence, $C$ is a one-rank hull code of length $n$  over $\mathbb{F}_{p^e},\  (e\geq3).$
\end{proof}
\begin{theorem}\label{one-rank hull over p^2}
	Let  $M$ be a duality over $\mathbb{F}_{p^2},\  (p\neq 2) $ with the duality matrix $D= \begin{pmatrix}
		a&b\\d&c
	\end{pmatrix}$. Then  a one-rank hull code of length one over $M$  exists if and only if $(b+d)^2-4ac$ is a quadratic residue modulo $p$.
\end{theorem}
\begin{proof}
	The quadratic  form $ \mathcal{Q} $ associated  with the duality matrix $D$ is $\mathcal{Q}(x,y)=ax^2+(b+d)xy+cy^2$. Suppose there exists  a one-rank hull code of length one over $M$. This implies that $\mathcal{Q}(x,y)=0$ has a non-trivial solution. Let $(x_0,y_0)$ be a non-trivial solution. If $x_0\neq 0$ and $y_0=0$ then $a=0$. If $x_0 = 0$ and $y_0\neq 0$ then $c=0$. In both cases, $(b+d)^2-4ac$ is a quadratic residue modulo $p$. Consider the case when $x_0$  and $y_0$ are non-zero. Now, we know that if $(x_0,y_0)$ is a solution of the equation, then $ (\lambda x_0,\lambda y_0)$ will be a solution for any $\lambda\in\mathbb{F}_p$. Hence, we can scale any solution $(x_0,y_0)$ to $(1,y_0')$  by multiplying $x_0^{-1}$. Thus, let $x_0=1$ then equation $\mathcal{Q}(x_0,y_0)=0$ reduces to one variable quadratic equation and the solution is  $y_0'=\frac{-(b+d)\pm\sqrt{(b+d)^2-4ac}}{2c}$. It follows that $(b+d)^2-4ac$ is a quadratic residue modulo $p$.
	 
	Conversely, suppose that $(b+d)^2-4ac$ is a quadratic residue modulo $p$. If either $ a=0 $ or $c=0$,  then $\mathcal{Q}(x,y)=0$ has a non-trivial solution of the type $(x_0,0)$ or $(0,y_0)$, respectively, for any $x_0,y_0\in \mathbb{F}_p$. Assume $a$ and $c$ are non-zero. Then $(x_0,\frac{-(b+d)\pm\sqrt{(b+d)^2-4ac}}{2c}x_0)$ will be a solution for any $x_0\in \mathbb{F}_p$.  Therefore, from Lemma \ref{selforthogonal}, we have self-orthogonal element $x\in \mathbb{F}_{p^2}$ such that $C=\left\langle x \right\rangle $ is a one-rank hull code of length one with respect to duality $M$ over $\mathbb{F}_{p^2}$.
\end{proof}
\begin{remark}
	Over $\mathbb{F}_4$, there are two dualities such that $\chi_x(x)=-1$ for all $x\in\mathbb{F}^*_4$. Therefore, one-rank hull codes of length one do not exist for theses dualities . 
\end{remark}
Warning's second theorem states that for any non-zero homogeneous polynomial $f$ of degree $d$ in $e$ variables over $\mathbb{F}_p$, the number of solutions $ (a_1,a_2,\dots,a_e)\in\mathbb{F}_p^e$ that satisfies $f(a_1,a_2,\dots,a_e)=0$ are at least $p^{e-d}$.
\begin{theorem}\label{existence of self-orthogonal}
	There exists at least $p^{e-2}$ self-orthogonal elements for any duality over $\mathbb{F}_{p^e}$.
\end{theorem}
\begin{proof}
	Let $ M $ be a duality over $\mathbb{F}_{p^e}$  with duality matrix $D$, and $\mathcal{Q}$ be a quadratic form of $M$. Applying Warning's second theorem to the quadratic form $\mathcal{Q}(\textbf{a}) = \textbf{a} D\textbf{a}^T$, implies that $\mathcal{Q}(\textbf{a}) = 0$ has at least $p^{e-2}$ solutions. Therefore, by Lemma \ref{selforthogonal}, we have at least $p^{e-2}$ self-orthogonal elements with respect to $M$ over $\mathbb{F}_{p^e}$.
\end{proof}
The next theorem gives the precise count of self-orthogonal elements with respect to all dualities over $\mathbb{F}_{p^e}$ for odd primes.
\begin{theorem}
	Let $M$ be any duality over $\mathbb{F}_{p^e}\ (p\neq 2)$ with duality matrix $D=[d_{ij}]$. Let $D'$ be a matrix of rank $k$ such that $D'=\left[ \frac{d_{ij}+d_{ji}}{2} \right] $. Then, we have the following.
	\begin{enumerate}
		\item If $k$ is odd, then there are exactly $p^{e-1}$ self-orthogonal elements corresponding to the duality $M$.
		\item If $k$ is even, then there are exactly $p^{e-k}\left(p^{k-1}+(p-1)p^{\frac{k-2}{2}}\eta\left( (-1)^{k/2}\Delta\right)  \right) $ self-orthogonal elements corresponding to the duality $M$, where $\eta$ is the quadratic character of $\mathbb{F}_p$ and $\Delta=$ non-zero minor of order $k$ of the equivalent diagonal form of $D'$.
	\end{enumerate}
	\end{theorem}
\begin{proof}
	Let $\mathcal{Q}$ be a quadratic form associated with the matrix $D'$. Then we have $\mathcal{Q}(\textbf{a})=\textbf{a}D'\textbf{a}^T=\textbf{a}D\textbf{a}^T$, for all $\textbf{a}\in \mathbb{F}^e_p$. The duality matrix $D$ is non-singular but not necessarily $D'$. Let $k$ be the rank of the matrix $D'$. From \cite[Theorem 6.21]{finitefields}, we know that $\mathcal{Q}$ is equivalent to some diagonal quadratic form, say $\mathcal{Q}'$. Now, the rank of the quadratic form $\mathcal{Q}'$ is $k$. Let $D''$  be the matrix after deleting   $e-k$ zero rows and columns from the coefficient matrix of $\mathcal{Q}'$. Define a non-degenerate quadratic form $\mathcal{Q}''$ with the coefficient matrix $D''$. We know that  the number of solutions of the equation $\mathcal{Q}(\textbf{a})=0$, is equal to $p^{e-k}$ times the number of solutions of the non-degenerate quadratic form $\mathcal{Q}''$ over $\mathbb{F}^k_p$.
	\begin{enumerate}
		\item If $k$ is odd, then by \cite[Theorem 6.27]{finitefields}, we have  exactly $p^{e-k}p^{k-1}=p^{e-1}$ number of solutions of the equation $\mathcal{Q}(\textbf{a})=0$ in $\mathbb{F}^e_p$. 
		\item If $k$ is even, then  by \cite[Theorem 6.26]{finitefields}, we have exactly $p^{e-k}\left(p^{k-1}+(p-1)p^{\frac{k-2}{2}}\eta\left( (-1)^{k/2}\Delta\right) \right) $ number of solutions of the equation $\mathcal{Q}(\textbf{a})=0$ in $\mathbb{F}^e_p$, where $\eta$ is the quadratic character of $\mathbb{F}_p$ and $\Delta=\det D''$.
	\end{enumerate}
	It is clear from Lemma \ref{selforthogonal} that every solution of the equation $\mathcal{Q}(\textbf{a})=0$ gives a self-orthogonal element corresponding to the duality $M$. Hence, the result follows.
\end{proof}
\begin{remark}
	If $M$ is a symmetric duality over $\mathbb{F}_{p^e}\ (p\neq 2) $ then $D'=D$. The duality matrix $D$ is the non singular matrix of rank $e$. If $e$ is odd, then there are exactly $p^{e-1}$ self-orthogonal elements corresponding to the duality $M$. If $e$ is even, then there are exactly $p^{e-1}+(p-1)p^{\frac{e-2}{2}}\eta\left( (-1)^{e/2}\Delta\right)$ self-orthogonal elements corresponding to the duality $M$, where $\eta$ is the quadratic character of $\mathbb{F}_p$ and $\Delta=\det D$. 
\end{remark}
\begin{exmp}
	Let $\mathbb{F}_{3^3}=\left\langle \nu_1,\nu_2,\nu_3\right\rangle $ be a finite field and $M$ be a duality over  $\mathbb{F}_{3^3}$ with duality matrix $D=\begin{pmatrix}
		1&2&0\\
		0&1&2\\
		1&0&1
	\end{pmatrix}$. Then the symmetric matrix $D'=\begin{pmatrix}
	1&1&2\\1&1&1\\2&1&1
	\end{pmatrix}$ such that the quadratic form $\mathcal{Q}(\textbf{a})=\textbf{a}D'\textbf{a}^T=\textbf{a}D\textbf{a}^T$, for all $\textbf{a}\in \mathbb{F}^3_3$. The rank of the matrix $D'$ is 3 (odd). One can observe that there are nine self-orthogonal elements, namely, 0, $2\nu_1+\nu_2,\  2\nu_2+\nu_3,\ 2\nu_1+\nu_2+2\nu_3,\ 2\nu_1+2\nu_3,\ \nu_1+2\nu_2,\ \nu_2+2\nu_3,\ \nu_1+2\nu_2+\nu_3,\ \nu_1+\nu_3$.
\end{exmp}
\begin{exmp}
	Let $M$ be a duality over $\mathbb{F}_{3^3}=\left\langle\nu_1,\nu_2,\nu_3 \right\rangle $ with duality matrix $D=\begin{pmatrix}
		2&1&1\\0&1&1\\1&0&1
	\end{pmatrix}$. Then symmetric matrix $D'=\begin{pmatrix}
	2&2&1\\2&1&2\\1&2&1
	\end{pmatrix}$ such that the quadratic form $\mathcal{Q}(\textbf{a})=\textbf{a}D'\textbf{a}^T=\textbf{a}D\textbf{a}^T$, for all $\textbf{a}\in \mathbb{F}^3_3$. The rank of the matrix $D'$ is $2$. We know that $\mathcal{Q}$ will be equivalent to a diagonal form; therefore, choose $C=\begin{pmatrix}
	1&0&0\\0&1&1\\2&0&1
\end{pmatrix}$ such that 
\begin{align*}
	C^TD'C=&\begin{pmatrix}
		1&0&2\\0&1&0\\0&1&1
	\end{pmatrix}\begin{pmatrix}
	2&2&1\\2&1&2\\1&2&1
	\end{pmatrix}\begin{pmatrix}
	1&0&0\\0&1&1\\2&0&1
	\end{pmatrix}\\
	=&\begin{pmatrix}
		1&0&0\\0&1&0\\0&0&0
	\end{pmatrix}.
\end{align*}
This implies that, $D''=\begin{pmatrix}
	1&0\\0&1
\end{pmatrix}$. Hence, there are $3(3+2\eta(-1))=3$ self-orthogonal elements, namely, 0, $\nu_2+\nu_3$ and $2\nu_2+2\nu_3$.
\end{exmp}
\section{Constructions of small rank hull codes}\label{Constructions}
The aim of this section is to give some construction techniques to construct small rank hull codes like ACD  and one-rank hull codes. We use ACD  and self-orthogonal codes to construct one-rank hull codes of different parameters.
\begin{theorem}
	Let $\mathbb{F}_{p^e}=\left\langle x_1,x_2,\dots,x_e \right\rangle $ be a finite field, and $M$ be any duality over $\mathbb{F}_{p^e}$. Then the additive code $C=\left\langle \bm{v_1},\bm{v_2},\dots, \bm{v_e}   \right\rangle $ is an  ACD code of length $n\ (\neq mp)$, where $\bm{v_i}=(x_i,x_i,\dots,x_i)$ for all $1\leq i\leq e$.
\end{theorem}
\begin{proof}
	Let $\textbf{u}\in C\cap C^M$, then $\textbf{u}=\sum_{i=1}^{e}n_i\bm{v_i}$ for some $n_i\in\mathbb{F}_p$ and $ \chi_{\textbf{u}}(\bm{v_j})=1$ for all $1\leq j\leq e$. It follows that $\prod_{i=1}^{e}(\chi_{\bm{v_i}}(\bm{v_j}))^{n_i}=\prod_{i=1}^{e}(\chi_{x_i}(x_j))^{nn_i}=\prod_{i=1}^{e}\xi^{k_{ij}n_in}=1$, where $\chi_{x_i}(x_j)=\xi^{k_{ij}}$ for any primitive $p$-th root of unity $\xi$ and $0 \leq k_{ij}\leq p-1$. This implies that $n\sum_{i=1}^{e}k_{ij}n_i\equiv 0 \mod p$ for all $1\leq j\leq e$. Since $n\not \equiv 0\mod p$ then $\sum_{i=1}^{e}k_{ij}n_i\equiv 0 \mod p$ for all $1\leq j\leq e$.  Hence, we have the  following system of equations
	$$\begin{pmatrix}
		k_{11}&k_{21}&\dots&k_{e1}\\
		k_{12}&k_{22}&\dots&k_{e2}\\
		\vdots&\vdots&\ddots&\vdots\\
		k_{1e}&k_{2e}&\dots&k_{ee}\\
	\end{pmatrix}
	\begin{pmatrix}
		n_1\\
		n_2\\
		\vdots\\
		n_e
	\end{pmatrix}\equiv 	\begin{pmatrix}
	0\\
	0\\
	\vdots\\
	0
	\end{pmatrix}\mod p
	$$ 
	In particular, we have
	$$D^T\begin{pmatrix}
		n_1\\
		n_2\\
		\vdots\\
		n_e
	\end{pmatrix}\equiv 	\begin{pmatrix}
	0\\
	0\\
	\vdots\\
	0
	\end{pmatrix}\mod p,
	$$ where $D$ is a duality matrix for the duality $M$. Note that the duality matrix is a non-singular matrix over $\mathbb{F}_p$. Therefore, this system of equations has only a trivial solution. Thus, the code $C$ is an ACD code.
\end{proof}
\begin{remark}
	In the above theorem, if $n=mp$, then $C$ is self-orthogonal code. Also note that $C$ is optimal ACD or self-orthogonal code with parameter $[n,p^e,n]$.
\end{remark}
\begin{exmp}
	Let $\mathbb{F}_{3^2}=\left\langle 1,\nu \right\rangle $. Then $C=\left\langle (1,1),(\nu,\nu)\right\rangle=\{(0,0),(1,1),(2,2),(\nu,\nu),(2\nu,2\nu),(1+\nu,1+\nu),(2+\nu,2+\nu),(1+2\nu,1+2\nu),(2+2\nu,2+2\nu)\}$ is an ACD code with parameter $[2,3^2,2]$ over $\mathbb{F}_{3^2}$ .
\end{exmp}
\begin{theorem}\label{ACD_from_selforthogonal}
	Let $\mathcal{G}$ be a generator matrix of  self-orthogonal $[n,p^k,d]$ additive code $C$ over $\mathbb{F}_{p^e}$. Let $x\in\mathbb{F}_{p^e}$ such that $\chi_x(x)\neq 1$. Then the code generated by the matrix $\mathcal{G}'=[xI|\mathcal{G}]$ is an $[n+k,p^k,d'>d]$ ACD code.
\end{theorem}
\begin{proof}
	Let $C'$ be an additive code generated by matrix $\mathcal{G}'=[xI|\mathcal{G}]$. Then the matrix $$\mathcal{G}'\odot_M\mathcal{G}'^T=\begin{pmatrix}
		\chi_x(x)\chi_{\mathcal{G}_1}(\mathcal{G}_1)&\chi_{\mathcal{G}_1}(\mathcal{G}_2)&\dots &\chi_{\mathcal{G}_1}(\mathcal{G}_k)\\
		\chi_{\mathcal{G}_2}(\mathcal{G}_1)&\chi_x(x)\chi_{\mathcal{G}_2}(\mathcal{G}_2)&\dots &\chi_{\mathcal{G}_2}(\mathcal{G}_k)\\
		\vdots&\vdots&\ddots&\vdots\\
		\chi_{\mathcal{G}_k}(\mathcal{G}_1)&\chi_{\mathcal{G}_k}(\mathcal{G}_2)&\dots &\chi_x(x)\chi_{\mathcal{G}_k}(\mathcal{G}_k)\\
	\end{pmatrix}.$$
	Since $C$ is self-orthogonal code, i.e., $C\subseteq C^M$ then $\chi_{\mathcal{G}_i}(\mathcal{G}_j)=1$ for all $i$ and $j$. Hence, we have $$\log_\xi(\mathcal{G}'\odot_M \mathcal{G}'^T)=\begin{pmatrix}
		\alpha&0&\dots&0\\
		0&\alpha&\dots&0\\
		\vdots&\vdots&\ddots&\vdots\\
		0&0&\dots&\alpha\\
	\end{pmatrix},$$ where $\chi_x(x)=\xi^{\alpha}, \alpha\neq 0$. Hence, by Theorem \ref{rankhull}, $C'$ is an  $[n+k,p^k,d'>d]$ ACD code.
\end{proof}
\begin{exmp}
	Let $C$ be an additive code over $\mathbb{F}_{3^2}$ with  generator matrix $\mathcal{G}=\begin{pmatrix}
		1&1&1&1&\nu\\
		\nu&\nu&2\nu&2\nu&1\\
		2&\nu&1&\nu&0
	\end{pmatrix}$.
	Observe that $C$ is a $[5,3^3,4]$ self-orthogonal additive code under the  symmetric duality $M_1$.  Choose $x=1$, since $\chi_1(1)=\xi\neq 1$. By applying Theorem \ref{ACD_from_selforthogonal}, we construct an $[8,3^3,5]$ ACD code and generator matrix is  $$\left( \begin{array}{ccc|ccccc}
	1&0&0&1&1&1&1&\nu\\
	0&1&0&\nu&\nu&2\nu&2\nu&1\\
	0&0&1&2&\nu&1&\nu&0
	\end{array}\right).$$
\end{exmp}

\begin{theorem}\label{onerank_from_selforthogonal}
	Let $\mathcal{G}$ be a generator matrix of  self-orthogonal $[n,p^k,d]$ additive code $C$. Let $x,y\in\mathbb{F}^*_{p^e}$ such that $\chi_x(x)\neq 1$ and $\chi_y(y)=1$. Then the code generated by the matrix $\mathcal{G}'=
	[diag(\underbrace{x,x,\dots,x}_\text{$k-1$ times},y)|\mathcal{G}]
	$ is an $[n+k,p^k,d'>d]$ one-rank hull code.
\end{theorem}
\begin{proof}
	Let $C'$ be an additive code generated by matrix $\mathcal{G}'$. Then the matrix $$\mathcal{G}'\odot_M\mathcal{G}'^T=\begin{pmatrix}
		\chi_x(x)\chi_{\mathcal{G}_1}(\mathcal{G}_1)&\chi_{\mathcal{G}_1}(\mathcal{G}_2)&\dots &\chi_{\mathcal{G}_1}(\mathcal{G}_k)\\
		\chi_{\mathcal{G}_2}(\mathcal{G}_1)&\chi_x(x)\chi_{\mathcal{G}_2}(\mathcal{G}_2)&\dots &\chi_{\mathcal{G}_2}(\mathcal{G}_k)\\
		\vdots&\vdots&\ddots&\vdots\\
		\chi_{\mathcal{G}_k}(\mathcal{G}_1)&\chi_{\mathcal{G}_k}(\mathcal{G}_2)&\dots &\chi_y(y)\chi_{\mathcal{G}_k}(\mathcal{G}_k)\\
	\end{pmatrix}$$
	Since $C$ is self-orthogonal code, i.e., $C\subseteq C^M$ then $\chi_{\mathcal{G}_i}(\mathcal{G}_j)=1$ for all $i$ and $j$. Hence we have $$\log_\xi(\mathcal{G}'\odot_M \mathcal{G}'^T)=\begin{pmatrix}
		\alpha&0&\dots&0&0\\
		0&\alpha&\dots&0&0\\
		\vdots&\vdots&\ddots&\vdots&\vdots\\
		0&0&\dots&\alpha&0\\
		0&0&\dots&0&0\\
	\end{pmatrix},$$ where $\chi_x(x)=\xi^{\alpha}, \alpha\neq 0$ and $\chi_y(y)=1$. Hence, by Corollary \ref{one-rank hull}, $C'$ is a one-rank hull code.
\end{proof}
\begin{exmp}
We start  with a $[5,3^3,4]$ self-orthogonal additive code under the duality $M_1$.  Choose $x=1$ and $y=1+\nu$, since $\chi_1(1)=\xi\neq 1$ and $\chi_{1+\nu}(1+\nu)=1$. By applying Theorem \ref{onerank_from_selforthogonal}, we construct an $[8,3^3,5]$ one-rank hull code and generator matrix is  $$\left( \begin{array}{ccc|ccccc}
	1&0&0&1&1&1&1&\nu\\
	0&1&0&\nu&\nu&2\nu&2\nu&1\\
	0&0&1+\nu&2&\nu&1&\nu&0
\end{array}\right).$$
\end{exmp}

 \begin{theorem}
 	Let $C$ be an additive code over $\mathbb{F}_{p^{2e}}$ and $M$ be any skew-symmetric duality. Let $\mathcal{G}$ be an $(2s+1)\times n$ generator matrix of $C$ such that $\chi_{\mathcal{G}_i}(\mathcal{G}_j)\neq 1 $, when $i$ and $j$ are consecutive integers; otherwise 1, where $\mathcal{G}_i$ denotes the $i$-th row of the matrix $\mathcal{G}$. Then, $C$ is a one-rank hull code. 
 	\end{theorem}
 	\begin{proof}
 		Since $M$ is a skew-symmetric duality then $\chi_{\mathcal{G}_i}(\mathcal{G}_i)=1$ and $\chi_{\mathcal{G}_j}(\mathcal{G}_i)=(\chi_{\mathcal{G}_i}(\mathcal{G}_j))^{-1}$ for all $i$ and $j$. Let $\chi_{\mathcal{G}_i}(\mathcal{G}_{(i+1)})=\xi_i$, where $\xi_i=\xi^{\alpha_i}$ and $\alpha_i\neq0$, for all $1\leq i\leq 2s$ and for some $\xi$ primitive p-th root of unity.
 		According to the theorem, we have 
 		 $$\mathcal{G}\odot_M\mathcal{G}^T=\begin{pmatrix}
 			1 & \xi_1 & 1&\dots&1&1\\
 			\xi_1^{-1}&1 & \xi_2&\dots &1&1\\
 			1&\xi_2^{-1}&1&\dots&1&1\\
 			\vdots&\vdots&\vdots&\ddots&\vdots&\vdots\\
 			1&1&1&\dots&1&\xi_{2s}\\
 			1&1&1&\dots&\xi_{2s}^{-1}&1
 		\end{pmatrix}.$$
 	 It follows that $$\log_\xi(\mathcal{G}\odot_M \mathcal{G}^T)=\begin{pmatrix}
 		0&\alpha_1&0&\dots&0&0\\
 		-\alpha_1&0&\alpha_2&\dots&0&0\\
 		0&-\alpha_2&0&\dots&0&0\\
 		\vdots&\vdots&\vdots&\ddots&\vdots&\vdots\\
 		0&0&0&\dots&0&\alpha_{2s}\\
 		0&0&0&\dots&-\alpha_{2s}&0
 		\end{pmatrix}.$$
 		We observe that the $rank(\log_\xi(\mathcal{G}\odot_M \mathcal{G}^T))=2s$. Hence, from Corollary \ref{one-rank hull}, the additive code generated by $\mathcal{G}$ is  a one-rank hull code.
 	\end{proof}
 	\begin{exmp}
Let $\mathcal{G}=\begin{pmatrix}
	1&1&1&1\\
	\nu&2\nu&2+\nu&0\\
	\nu&1+2\nu&2\nu&\nu
\end{pmatrix}$	be a generator matrix of an additive code $C$ over $\mathbb{F}_{3^2}$. Consider the skew-symmetric duality $M_2$ over $\mathbb{F}_{3^2}$. Then one can deduce that $\mathcal{G}\odot_M\mathcal{G}^T=\begin{pmatrix}

1&\xi&1\\
\xi^2&1&\xi^2\\
1&\xi&1
\end{pmatrix}$ and $\log_\xi(\mathcal{G}\odot_M \mathcal{G}^T)=\begin{pmatrix}
0&1&0\\
2&0&2\\
0&1&0
\end{pmatrix}$. The rank of the matrix $\log_\xi(\mathcal{G}\odot_M \mathcal{G}^T)$ is 2, hence from Corollary \ref{one-rank hull}, $C$ is a one-rank hull code.
 	\end{exmp}
 	\begin{theorem}\label{ACD_to_onerank}
 		Let $\mathcal{G}$ be a $ 2s \times n $ generator matrix of an ACD code $C$ over $\mathbb{F}_{p^{2e}}$ and $M$ be any skew-symmetric duality. Let $\textbf{x}\notin C$ then the additive code generated by a matrix $\mathcal{G}'=\begin{pmatrix}
 			\textbf{x}\\
 			\mathcal{G}
 		\end{pmatrix}$  is q one-rank hull code.
 	\end{theorem}
 	\begin{proof}
 		Let $C'$ be an additive code with generator matrix $\mathcal{G}'=\begin{pmatrix}
 			\textbf{x}\\
 			\mathcal{G}
 		\end{pmatrix}$. Since $M$ is a skew-symmetric duality then $\chi_{\textbf{x}}(\textbf{x})=1$. Let $\chi_{\textbf{x}}(\mathcal{G}_i)=\xi^{\alpha_i}$ then $\chi_{\mathcal{G}_i}(\textbf{x})=\xi^{-\alpha_i}$ for all $i$. Hence, we have $$\mathcal{G}'\odot_M\mathcal{G}'^T=\begin{pmatrix}
 		1& \chi_{\textbf{x}}(\mathcal{G}_1)&\dots&\chi_{\textbf{x}}(\mathcal{G}_{2s})\\
 		
 		\chi_{\mathcal{G}_1}(\textbf{x})& & &\\
 		
 		\vdots& &\mathcal{G}\odot_M\mathcal{G}^T\\
 		\chi_{\mathcal{G}_{2s}}(\textbf{x})
 		\end{pmatrix}$$ and $$\log_\xi(\mathcal{G}'\odot_M \mathcal{G}'^T)=\begin{pmatrix}
 		0&\alpha_1&\cdots&\alpha_{2s}\\
 		-\alpha_1& & &\\
 		\vdots & & \log_\xi(\mathcal{G}\odot_M \mathcal{G}^T)\\
 		-\alpha_{2s}& & &\\
 		\end{pmatrix}_{(2s+1)\times(2s+1)}.$$ Then the matrix  $\log_\xi(\mathcal{G}'\odot_M \mathcal{G}'^T)$ is odd order skew-symmetric matrix. Hence the rank of the matrix is less then or equal to $2s$. Since $C$ is an ACD code therefore $\log_\xi(\mathcal{G}\odot_M \mathcal{G}^T)$ is a non-singular matrix of order $2s$.  Thus, the rank of the matrix $\log_\xi(\mathcal{G}'\odot_M \mathcal{G}'^T)$ is $ 2s$. By Corollary \ref{one-rank hull}, $C'$ is a one-rank hull code.
 	\end{proof}
 	\begin{exmp}
 		We start with $[4,3^2,2]$ ACD code  $C$ with respect to skew-symmetric duality $M_2$ over $\mathbb{F}_{3^2}$. The generator matrix of $C$ is $\begin{pmatrix}
 			1&1&0&0\\
 			\nu&\nu&\nu&\nu
 		\end{pmatrix}$. Let $\textbf{x}=(\nu,\nu,1,1)\notin C$.  By Theorem \ref{ACD_to_onerank}, we construct $[4,3^3,2]$ one-rank hull code and the generator matrix is $$\begin{pmatrix}
 			\nu&\nu&1&1\\
 			1&1&0&0\\
 			\nu&\nu&\nu&\nu
 		\end{pmatrix}.$$
 	\end{exmp}
\begin{theorem}\label{ACD_to_n+1onerank}
	Let $\mathcal{G}$ be a $2s\times n$ generator matrix of an $[n,p^{2s},d]$ ACD code $C$ over $\mathbb{F}_{p^{2e}}$ and $M$ be any skew-symmetric duality. Then for any $\textbf{x}\notin C$ and $\alpha\in\mathbb{F}^*_{p^{2e}}$, the additive code generated by a matrix $\mathcal{G}'=\begin{pmatrix}
		\alpha& \textbf{x}\\
		\alpha& \mathcal{G}_1\\
		\vdots &\vdots\\
		\alpha& \mathcal{G}_{2s}
	\end{pmatrix}$ is $[n+1,p^{2s+1}]$ one-rank hull code, where $\mathcal{G}_i's$ are the $i$-th row of the matrix $\mathcal{G}$.
\end{theorem}
\begin{proof}
	Let $C'$ be an additive code generated by matrix $\mathcal{G}'$. Then the matrix $$\mathcal{G}'\odot_M\mathcal{G}'^T=\begin{pmatrix}
		1& \chi_{\textbf{x}}(\mathcal{G}_1)&\dots&\chi_{\textbf{x}}(\mathcal{G}_{2s})\\
		
		\chi_{\mathcal{G}_1}(\textbf{x})& & &\\
		
		\vdots& &\mathcal{G}\odot_M\mathcal{G}^T\\
		\chi_{\mathcal{G}_{2s}}(\textbf{x})
	\end{pmatrix}.$$
It follows that the matrix $\log_\xi(\mathcal{G}'\odot_M \mathcal{G}'^T) $ is odd order skew-symmetric matrix having a non-zero minor of order $2s$. Therefore, by Corollary \ref{one-rank hull}, $C'$ is $[n+1,p^{2s+1}]$ one-rank hull code.
\end{proof}
\begin{exmp}
	We start with $[4,3^2,2]$ ACD code  $C$ with respect to skew-symmetric duality $M_2$ over $\mathbb{F}_{3^2}$. The generator matrix of $C$ is $\begin{pmatrix}
	1&1&0&0\\
	\nu&\nu&\nu&\nu
\end{pmatrix}$. Choose $\textbf{x}=(\nu,1,1,1)\notin C$ and $\alpha=1$.  By Theorem \ref{ACD_to_n+1onerank}, we construct $[5,3^3,3]$ one-rank hull code and the generator matrix is $$\left( \begin{array}{c|cccc}
	1&\nu&1&1&1\\
	1&1&1&0&0\\
	1&\nu&\nu&\nu&\nu
\end{array}\right) .$$
\end{exmp}
\section{Optimal additive codes with one-rank hull}\label{Bounds}
In this section, we study the highest possible minimum distance among all additive codes with one-rank hull. To pave the way for this, we will first discuss some relevant results that will be helpful in our analysis.
\begin{lemma}\label{xi^s if q(u)=s}
	Let $M$ be any duality over $\mathbb{F}_{p^e}=\left\langle x_1,x_2,\dots,x_e \right\rangle $  and $\mathcal{Q}$ be a quadratic form associated with duality matrix $D$. Then for any $u\in \mathbb{F}_{p^e}$, $\chi_u(u)=\xi^s$ if and only if $\mathcal{Q}(u_1,u_2,\dots, u_e)=s$, where $u=\sum_{i=1}^{e}u_ix_i$ and $s\in\mathbb{F}_p$.
\end{lemma}
\begin{proof}
		We have $ \chi_u(u)=\chi_{\sum_{i=1}^{e}u_ix_i}(\sum_{i=1}^{e}u_ix_i)=\prod_{i,j=1}^{e}\chi_{x_i}(x_j)^{u_iu_j}=\xi^{\sum_{i,j=1}^{e}k_{ij}u_iu_j}=\xi^s$ if and only if  $\sum_{i,j=1}^{e}k_{ij}u_iu_j\equiv s\mod p$. Also,  $\mathcal{Q}(u_1,u_2,\dots, u_e)=[u_1,u_2,\dots, u_e]D[u_1,u_2,\dots, u_e]^T=\sum_{i,j=1}^{e}k_{ij}u_iu_j$. Hence, the result follows.
		
\end{proof}
\begin{theorem}\label{existence of u for any s over Fp2}
Let $M$ be a duality over $\mathbb{F}_{p^2}\ (p\neq 2),$ such that $\chi_u(u)\neq 1$ for all $u\in\mathbb{F}^*_{p^2}$. Then for each $s\in\mathbb{F}^*_p$ there exists $u\in\mathbb{F}^*_{p^2} $ such that $\chi_u(u)=\xi^s$. Moreover, there are $(p+1)$ elements of $\mathbb{F}^*_{p^2} $ with $\chi_u(u)=\xi^s$ for each $s\in\mathbb{F}^*_p$.
	\end{theorem}
	\begin{proof}
		Let $D=\begin{pmatrix}
			a&b\\d&c
		\end{pmatrix}$ be a duality matrix of the duality $M$ over $\mathbb{F}_{p^2},\ (p\neq2)$. Since there is no non-zero self-orthogonal element, then by Theorem \ref{one-rank hull over p^2}, we have $(b+d)^2-4ac\not\equiv0\mod p$. Let $D'=\begin{pmatrix}
		a&(b+d)/2\\(b+d)/2&c
		\end{pmatrix}$ be a matrix, then $rank(D')=2$. Let $\mathcal{Q}$ be a  quadratic form  associated with $D'$. By \cite[Theorem 6.21]{finitefields}, quadratic form $\mathcal{Q}$ is equivalent to a diagonal quadratic form, say $\mathcal{Q}'$. Hence, $\mathcal{Q}$ and $\mathcal{Q}'$  have equal number of solutions. Let $\mathcal{Q}'(u_1,u_2)=a_1u^2_1+a_2u^2_2$. Then by \cite[Lemma 6.24]{finitefields}, we have $p+(-1)\eta(-a_1a_2)$ number of solutions of $a_1u^2_1+a_2u^2_2=s$ for each $s\in\mathbb{F}^*_p$. Since there does not exist any non-zero solution for $a_1u^2_1+a_2u^2_2=0$, therefore $\eta(-a_1a_2)=-1$. Thus, by  Lemma \ref{xi^s if q(u)=s}, we have $(p+1)$ elements of  $\mathbb{F}^*_{p^2} $ with $\chi_u(u)=\xi^s$ for each $s\in\mathbb{F}^*_p$.
	\end{proof}
\begin{corollary}\label{xi^a then xi^-a}
	Let $M$ be a duality over $\mathbb{F}_{p^2}\ (p\neq 2),$ such that $\chi_u(u)\neq 1$ for all $u\in\mathbb{F}^*_{p^2}$. If $\chi_u(u)=\xi^s$ for any $s\in\mathbb{F}^*_p$ then there exists $v\in\mathbb{F}^*_{p^2}$ such that $\chi_v(v)=\xi^{-s}$.
\end{corollary}
\begin{proof}
	If $s\in\mathbb{F}^*_p$, then $-s\in\mathbb{F}^*_p$. The proof follows directly from the above Theorem \ref{existence of u for any s over Fp2}.
\end{proof}
\begin{remark}
	The above corollary holds for $p=2$ since we have $s=1$ only. 
\end{remark}
\begin{theorem}\label{chi u(v)=1}
	For each $u\in\mathbb{F}_{p^e}$, there  exists $v\in \mathbb{F}_{p^e}$ such that $\chi_u(v)=1$ with respect to any duality $M$. Moreover, there are $p^{e-1}$ elements of $\mathbb{F}_{p^e}$ such that $\chi_u(v)=1$ for each $u \in \mathbb{F}^*_{p^e}$.
\end{theorem}
\begin{proof}
Let $u=\sum_{i=1}^{e}n_ix_i$ and 	 $v=\sum_{i=1}^{e}m_ix_i$, where $x_1,x_2,\dots,x_e$ are generators of $ \mathbb{F}_{p^e} $ and $n_i,m_i\in \mathbb{F}_{p}$ for all $1\leq i\leq e$. For any fixed $u$, $\chi_u(v)=1$ if and only if $$([n_1,n_2,\dots,n_e]D)[m_1,m_2,\dots,m_e]^T\equiv 0\mod p,$$ where $D$ is a duality matrix. The matrix $[n_1,n_2,\dots,n_e]D$ has rank at most one and nullity at least $e-1$. Moreover, for non-zero $u$, rank of the matrix $[n_1,n_2,\dots,n_e]D$ is one and nullity $e-1$. Therefore, there are $p^{e-1}$ elements of $\mathbb{F}_{p^e}$ such that $\chi_u(v)=1$ for each $u \in \mathbb{F}^*_{p^e}$.
\end{proof}
\begin{remark}
	Note that if $e=2$ and $\chi_u(u)=1$, for some $u \in \mathbb{F}^*_{p^2}$, then $\chi_u(v)\neq 1$ for all $v\neq \alpha u$, where $\alpha \in \mathbb{F}_{p}$.
\end{remark}
The notation $[n, p^k, d]$ represent an additive code with parameters of length $n$, cardinality $p^k$, and distance $d$ over the finite field $ \mathbb{F}_{p^e}$. We introduce the concept of one-rank hull codes over $ \mathbb{F}_{p^e}$ in relation to their highest possible distance, denoted as $d_{1}[n, k]_{p^e,M}$. It is defined as $$d_{1}[n, k]_{p^e,M}=\max\{d\ |\ \text{$\exists$ an $[n, p^k, d]$ one-rank hull code with respect to the duality $M$}\}.$$
\begin{theorem}
 For any duality $M$ over $\mathbb{F}_{p^e}$, where $e\geq 3$, we have $d_{1}[n, 1]_{p^e,M}=n$.
\end{theorem}
\begin{proof}
	From Theorem \ref{existence of self-orthogonal}, we have an element $u\in\mathbb{F}^*_{p^e}$ such that $\chi_u(u)=1$, for any duality $M$ and, for any finite field  $\mathbb{F}_{p^e}$, where $e\geq 3$. Therefore, a code $C=\left\langle (u,u,\dots,u)\right\rangle $ is a one-rank hull code with distance $n$. Hence, the result follows.
\end{proof}
There are dualities over $\mathbb{F}_{p^2}$ such that there is no non-zero self-orthogonal element. In this case, there is no one-rank hull code of length 1. However, if $\chi_u(u)=1$ for some $u\in\mathbb{F}^*_{p^2}$ then  for such dualities $d_1[1,1]=1$. In the following theorem, we find  $d_1[n,1]$ over $\mathbb{F}_{p^2}$ for $n\geq 2$. 
\begin{theorem}
	For any duality $M$ over $\mathbb{F}_{p^2}$, where $p\neq 2$, we have $d_1[n,1]_{p^2,M}=n\ (n\geq 2)$.
\end{theorem}
\begin{proof}
	Let $M$ be a duality over $\mathbb{F}_{p^2}$ and  $u\in\mathbb{F}^*_{p^2}$ such that $\chi_u(u)=1$. In this case, $C=\left\langle (u,u,\dots,u)\right\rangle $ is a one-rank hull code with distance $n$. Now, consider the duality such that there is no $u\in\mathbb{F}^*_{p^2}$ with $\chi_u(u)=1$. Let $\chi_u(u)=\xi^s$ for some $1\leq s\leq p-1$.\\
				\textbf{Case 1:} If $n=mp+1$, for any   integer $m>0$, then there exists a vector  $$\textbf{y}=(\underbrace{u,u,\dots,u}_{(m-1)p\ \text{times}},\underbrace{u,u,\dots,u}_{(p+1)/2\ \text{times}},\underbrace{v,v,\dots,v}_{(p+1)/2\ \text{times}}),$$ such that $\chi_\textbf{y}(\textbf{y})=(\chi_u(u))^{(m-1)p}(\chi_u(u))^{(p+1/2)}(\chi_v(v))^{(p+1/2)}=1\cdot\xi^{s((p+1)/2)}\xi^{-s((p+1)/2)}=1.$  Since $s\not \equiv 0\mod p$, by Corollary \ref{xi^a then xi^-a}, we have an element $v\in\mathbb{F}^*_{p^2}$ such that $\chi_v(v)=\xi^{-s}$.\\
		\textbf{Case 2:} If $n\neq mp+1$, for any  integer  $m\geq0$  then there exists $$\textbf{y}=(\underbrace{u,u,\dots,u}_{n-1},w),$$ such that $\chi_\textbf{y}(\textbf{y})=(\chi_u(u))^{n-1}(\chi_w(w))=\xi^{s(n-1)}\xi^{-s(n-1)}=1.$  Since $s(n-1)\not \equiv 0\mod p$, by Corollary \ref{xi^a then xi^-a}, we have an element $w\in\mathbb{F}^*_{p^2}$ such that $\chi_w(w)=\xi^{-s(n-1)}$.\\
	In each case, $C=\left\langle \textbf{y} \right\rangle $ is the one-rank hull code with distance $n$. Hence, the result follows.
\end{proof}
\begin{theorem}\label{F4_1}
For  symmetric dualities over $\mathbb{F}_4$, we have  $d_{1}[n, 1]_{4,M}=n$,  and for non-symmetric dualities,  $d_{1}[n, 1]_{4,M}=\begin{cases}
		n & \ \text{if $n$ is even},\\
		n-1 &\  \text{if $n$ is odd}
	\end{cases}$.
\end{theorem}
\begin{proof}
	Over $\mathbb{F}_4$, for symmetric dualities there exists $u\in\mathbb{F}^*_4$ such that $\chi_u(u)=1$. In this case $C=\left\langle (u,u,\dots,u)\right\rangle $ is a one-rank hull code with distance $n$. For non-symmetric dualities, we have $\chi_u(u)=-1$ for all $u\in\mathbb{F}^*_4$. If $n$ is even, then any additive code $C=\left\langle (v_1,v_2,\dots,v_n)\right\rangle $ with distance $n$ is a one-rank hull code. Since $\prod_{i=1}^{n}\chi_{v_i}(v_i)=(-1)^n=1$.  If $n$ is odd, then there is no vector $\textbf{v}=(v_1,v_2,\dots,v_n)$ of weight $n$ such that $\chi_\textbf{v}(\textbf{v})=1$. Since $\chi_\textbf{v}(\textbf{v})=\prod_{i=1}^{n}\chi_{v_i}(v_i)=(-1)^n=-1$. Thus, there is one-rank hull code $C=\left\langle (v_1,v_2,\dots,v_{n-1},0)\right\rangle $ with distance $n-1$. Hence, the result holds.
\end{proof}

The Singleton bound for an $[n, p^k, d]$ additive code over $\mathbb{F}_{p^e}$ states that the minimum distance $d$ of the code is bounded by $d\leq n-\lceil \frac{k}{e}\rceil +1$.  Additionally, two vectors $\textbf{x}$ and $\textbf{y}$ in $\mathbb{F}_{p^e}$  are considered $p$-linearly independent if they are linearly independent over $\mathbb{F}_p$.
\begin{theorem}\label{ne-1}
	 For all dualities $M$ over $\mathbb{F}_{p^e}$ for $n\geq 2$, we have $d_1[n,ne-1]_{p^e,M}=1$.
\end{theorem}
\begin{proof}
From the Singleton bound, we have $d_1[n,ne-1]_{p^e,M}\leq 1$. Now, if $C$ is an $[n,p^1]$ one-rank hull  code then 	 $C^{M^T}$ is an $[n,p^{ne-1}]$ one-rank hull code  with respect to duality $M$ (see Proposition \ref{C^{M^T} is one-rank}). Hence, 	$d_1[n,ne-1]_{p^e,M}=1$, for all dualities $M$ over $\mathbb{F}_{p^e}$ for $n\geq 2$. 
\end{proof}
It is worth mentioning that when $k=2$, there is no  one-rank hull code with respect to skew-symmetric dualities (see Theorem \ref{skew-symmetric k odd}). As a result, the following theorems  will focus on dualities which are not skew-symmetric.
\begin{theorem}
		 For  any duality $M$ over $\mathbb{F}_{p^e}$, where $p\neq 2$ and $e\geq3$, we have $d_1[n,2]_{p^e,M}=n\ (n\geq2)$.
\end{theorem}
\begin{proof}
	It is clear from Theorem \ref{existence of 1 rank for geq 3} that there exists $u\in\mathbb{F}^*_{p^e}$ such that $\chi_u(u)=1$. Since we are considering dualities that are not skew-symmetric, then we  have an element $w\in\mathbb{F}^*_{p^e}$ with $\chi_w(w)=\xi^s\neq1$. Furthermore, according to Theorem \ref{chi u(v)=1}, there is an element $v$ in $\mathbb{F}^*_{p^e}$ for which  $\chi_u(v)=1$ and $v\neq \alpha u$ for any $\alpha\in\mathbb{F}_p$. This is significant because there are  $p^{e-1}(>p)$ elements such that $\chi_u(v)=1$.\\
	\textbf{Case 1:} If $\chi_v(v)\neq 1$ and $n\neq mp$ for any positive integer $m$. Then choose  $$\mathcal{G}=\begin{pmatrix}
			u&u&\cdots&u\\
			v&v&\cdots&v
		\end{pmatrix}.$$\\
	\textbf{Case 2:} If $\chi_v(v)\neq 1$ and $n=mp$ for some positive integer $m$. Then choose $$\mathcal{G}=\left( \begin{array}{cccccc}
			u&u&\cdots&u&u-v&u+v\\
			\undermat{mp-2}{v&v&\cdots&}v&v&v
		\end{array}\right) .$$  \\
		\textbf{Case 3:} If $\chi_v(v)= 1$. Then choose $$\mathcal{G}=\left( \begin{array}{ccccc}
			u&u&u&\cdots&u\\
			w&-w&\undermat{n-2}{v&\cdots&v}
		\end{array}\right) .$$ \\
		In each case, an additive code generated by $\mathcal{G}$ is a one-rank hull code with distance $n$. Since the matrix $\log_\xi(\mathcal{G}\odot_M\mathcal{G}^T)$  is 1 rank matrix of the type $\begin{pmatrix}
			0&0\\
			*&a
		\end{pmatrix}$,  $\begin{pmatrix}
		a&0\\
		*&0
		\end{pmatrix}$ and $\begin{pmatrix}
			0&0\\
			*&a
		\end{pmatrix}$, respectively, where $a\neq 0$ and $*$ could be  0 or non-zero entry. 
\end{proof}
\begin{theorem}
  For  any duality $M$ over $\mathbb{F}_{2^e}$, where $e\geq3$, we have	$d_1[n,2]_{2^e,M}=n$.
\end{theorem}
\begin{proof}
	It is clear from Theorem \ref{existence of 1 rank for geq 3} that there exists $x\in\mathbb{F}^*_{2^e}$ such that $\chi_x(x)=1$. Further, according to Theorem \ref{chi u(v)=1}, there exists $y_1,y_2,\cdots,y_{e-2}\in\mathbb{F}^*_{2^e}$ such that $\chi_x(y_i)=1$ for all $1\leq i \leq e-2$ and $\{x,y_1,y_2,\cdots,y_{e-2}\}$ is a set of $2$-linear independent elements. We can extend this set to a $2$-linear independent set of $\mathbb{F}_{2^e} $, say $\{x,y_1,\cdots,y_{e-2},z\}$. Thus, $\chi_x(z)=-1$. 
	
	To prove $d_1[n,2]_{2^e,M}=n$, it is enough to give a self-orthogonal code of length $2$, a one-rank hull code of length 1 and a one-rank hull code of length 2. It is clear to see that $\begin{pmatrix}
		x&x\\z&z
	\end{pmatrix}$ generates a self-orthogonal code of length 2.
	
	Now, if $\chi_{y_i}(x)=-1$ for some $i$, then $\mathcal{G}_1=\begin{pmatrix}
		x\\y_i
	\end{pmatrix}$ and $\mathcal{G}_2=\begin{pmatrix}
	x&x\\y_i&x+y_i
	\end{pmatrix}$ generates one-rank hull codes of length 1 and 2, respectively. Therefore, we assume $\chi_{y_i}(x)=1$ for all $i$.\\
	\textbf{Case 1:} If $\chi_{y_i}(y_i)=-1$ for some $i$, then $\mathcal{G}_1=\begin{pmatrix}
		x\\y_i
	\end{pmatrix}$ generates a one-rank hull code of length one. Either  $\begin{pmatrix}
	x+y_i&x\\z&z
	\end{pmatrix}$ or $\begin{pmatrix}
	x&y_i\\z&z
	\end{pmatrix}$ generates a one-rank hull code depending on the value of $\chi_{y_i}(z)$.\\
	\textbf{Case 2:} If $\chi_{y_i}(y_i)=1$ for all $i$. In this case, $\begin{pmatrix}
		x\\y_i
	\end{pmatrix}$ generates a self-orthogonal code of length one. Thus, it is enough to construct a one-rank hull code of length one.
	\begin{enumerate}
	\item Let $\chi_z(x)=1$, then $\mathcal{G}_1=\begin{pmatrix}
		x\\z
	\end{pmatrix}$ generates a one-rank hull code.
	\item Let  $\chi_z(x)=-1$. If $\chi_{y_i}(z)\chi_z(y_i)=-1 $, for some $i$, then $\mathcal{G}_1=\begin{pmatrix}
		z\\x+y_i
	\end{pmatrix}$ gives the required code. Thus, assume $\chi_{y_i}(z)\chi_z(y_i)=1$ for all $i$. We claim that $\chi_z(z)$ must be -1. Let $w=ax+bz+\sum_{i=1}^{e-2}n_iy_i$ and $\chi_{z}(z)=1$. Then \begin{align*}
\chi_w(w)=&(\chi_x(z))^{ab}(\chi_z(x))^{ab}\left( \chi_{z}( \sum_{i=1}^{e-2}n_iy_i) \right) ^b\left( \chi_{\sum_{i=1}^{e-2}n_iy_i}(z)\right) ^b\\
=&(-1)^{ab}(-1)^{ab}\left( \prod_{i=1}^{e-2}(\chi_{y_i}(z)\chi_{z}(y_i))^{n_i}\right)^b\\
&=1. 
	\end{align*} This implies that $M$ is a skew-symmetric duality, which is a contradiction. Thus, $\chi_z(z)=-1$. Hence, either $\begin{pmatrix}
	z\\y_i
	\end{pmatrix}$ or $\begin{pmatrix}
	z\\x+y_i
	\end{pmatrix}$ generates a one-rank hull code depending on the value of $\chi_z(y_i) $.
	\end{enumerate}
	
\end{proof}
Note that there is no one-rank hull code of length $1$ and $k=2$ for any duality over $\mathbb{F}_{p^2}$.  
\begin{theorem}
For all dualities $M$, where $\chi_u(u)=1$ for some $u\in\mathbb{F}^*_{p^2}$ and $p\neq2$, we have	$d_1[n,2]_{p^2,M}=n$.
\end{theorem}
\begin{proof}
	Let $\chi_u(u)=1$ for some $u\in\mathbb{F}^*_{p^2}$. We are taking  dualities that are not skew-symmetric, therefore, there exists $v\in\mathbb{F}^*_{p^2}$ such that $\chi_v(v)=\xi^s\neq 1$, for some $1\leq s\leq p-1$. Also, from Theorem \ref{chi u(v)=1},   $\chi_u(v)\neq 1$.\\
	\textbf{Case 1:} If  $n=mp$, for some positive integer $m$. Then choose   $$\mathcal{G}=\begin{pmatrix}
		u-v&u+v&u&\cdots&u\\
		v&v&\undermat{(n-2)}{v&\cdots&v}
	\end{pmatrix}.$$  \\
	\textbf{Case 2:}   If $n=mp+1$, for any   integer $m>0$.  Then choose   $$\mathcal{G}=\left( \begin{array}{ccccccccc}
		u&\cdots&u&u&\cdots&u&u&\cdots&u\\
		\undermat{(p+1)/2}{-v&\cdots&-v}&\undermat{(p+1)/2}{v&\cdots&v}&\undermat{(m-1)p}{v&\cdots&v}
	\end{array}\right). $$ \\
	\textbf{Case 3:} If $n=mp+k$, for any  integer  $m\geq0$ and $2\leq k\leq p-1$. Then choose   $$\mathcal{G}=\begin{pmatrix}
		u&u&\cdots&u\\
		-(n-1)v&\undermat{(n-1)}{v&\cdots&v}
	\end{pmatrix}.$$ \\
	In each case, an additive code generated by $\mathcal{G}$ is a one-rank hull code with distance $n$. Since the matrix $\log_\xi(\mathcal{G}\odot_M\mathcal{G}^T)$  is 1 rank matrix of the type $\begin{pmatrix}
		a&*\\
		0&0
	\end{pmatrix}$, $\begin{pmatrix}
		0&0\\
		0&a
	\end{pmatrix}$ and $\begin{pmatrix}
		0&0\\
		0&a
	\end{pmatrix}$, respectively, where $a\neq 0$ and $*$ could be  0 or non-zero entry.
\end{proof}
\begin{theorem}
	 For all dualities $M$, where  $\chi_u(u)\neq1$ for all $u\in\mathbb{F}^*_{p^2}$ and $p\neq2,3$, we have $d_1[n,2]_{p^2,M}=n$.
\end{theorem}
\begin{proof}
		Let $u\in\mathbb{F}^*_{p^2}$ such that $\chi_u(u)=\xi^s$, for some $1\leq s\leq p-1$. From Corollary \ref{xi^a then xi^-a}, we have an element $w\in\mathbb{F}^*_{p^2}$ such that $\chi_w(w)=\xi^{-s}$. Now, from Theorem \ref{chi u(v)=1}, there exists $v,z\in\mathbb{F}^*_{p^2}$ such that $\chi_u(v)=1$ and $\chi_w(z)=1$, where $\{u,v\}$ and $\{w,z\}$ are $p$-linear independent sets. Let $\chi_v(v)=\xi^t$, for some $1\leq t\leq p-1$.\\
		
				\textbf{Case 1:}   Let $n=mp+1$, for any   integer $m>0$.
			\begin{enumerate}
				\item If $\chi_z(z)\neq\xi^{-t}$, then choose   $$\mathcal{G}=\left( \begin{array}{ccccccccc}
					w&\cdots&w&	u&\cdots&u&u&\cdots&u\\
					\undermat{(p+1)/2}{z&\cdots&z}&	\undermat{(p+1)/2}{v&\cdots&v}&\undermat{(m-1)p}{v&\cdots&v}
				\end{array}\right).$$ 
				\item  If $\chi_z(z)=\xi^{-t}$, then choose   $$\mathcal{G}=\left( \begin{array}{cccccccccc}
					w&w&\cdots&w&	u&\cdots&u&u&\cdots&u\\
					\undermat{(p+1)/2}{2z&z&\cdots&z}&	\undermat{(p+1)/2}{v&\cdots&v}&\undermat{(m-1)p}{v&\cdots&v}
				\end{array}\right). $$\\
			\end{enumerate}
		
		\textbf{Case 2:} Let $n\neq mp+1$, for any  integer  $m\geq0$.  Since $s(n-1)\neq 0\mod p$, by Corollary \ref{xi^a then xi^-a}, we have an element $y\in\mathbb{F}^*_{p^2}$ such that $\chi_y(y)=\xi^{-s(n-1)}$. Also, from Theorem \ref{chi u(v)=1}, there exists $x\in\mathbb{F}^*_{p^2}$ such that $\chi_y(x)=1$. 
		\begin{enumerate}
			\item If $\chi_x(x)\neq \xi^{-t(n-1)}$ then
			 choose   $$\mathcal{G}=\begin{pmatrix}
				y&u&\cdots&u\\
				x&\undermat{(n-1)}{v&\cdots&v}
			\end{pmatrix}.$$
			\item If $\chi_x(x)=\xi^{-t(n-1)}$ then
			choose   $$\mathcal{G}=\begin{pmatrix}
				y&u&\cdots&u\\
				2x&\undermat{(n-1)}{v&\cdots&v}
			\end{pmatrix}.$$
		\end{enumerate}  
		In each case, an additive code generated by $\mathcal{G}$ is a one-rank hull code with distance $n$. Since the matrix $\log_\xi(\mathcal{G}\odot_M\mathcal{G}^T)$  is 1 rank matrix of the type $\begin{pmatrix}
			0&0\\
			*&a
		\end{pmatrix}$, where $a\neq 0$ and $*$ could be  0 or non-zero entry.
\end{proof}
\begin{theorem}
  For all dualities $M$, where  $\chi_u(u)\neq1$, for all $u\in\mathbb{F}^*_{9}$, we have	$d_1[n,2]_{9,M}=n\ (n\geq 3)$.
\end{theorem}
\begin{proof}
		Let $u\in\mathbb{F}^*_{9}$ such that $\chi_u(u)=\xi$. From Corollary \ref{xi^a then xi^-a}, we have an element $w\in\mathbb{F}^*_{9}$ such that $\chi_w(w)=\xi^{2}$. Now, from Theorem \ref{chi u(v)=1}, there exists $z\in\mathbb{F}^*_{9}$ such that  $\chi_w(z)=1$, where  $\{w,z\}$ is $3$-linear independent set. \\
	\textbf{Case 1:} Let  $n=3m$ for some  integer $m>0$. Then choose   $$\mathcal{G}=\begin{pmatrix}
		u-w&u+w&u&\cdots&u\\
		w&w&\undermat{n-2}{w&\cdots&w}
	\end{pmatrix}.$$  \\
	\textbf{Case 2:} Let  $n=3m+4$ for some  integer $m\geq 0$. Then choose   $$\mathcal{G}=\left( \begin{array}{ccccccc}
		u&u&w&w&u&\cdots&u\\
	    (u-w)&-(u-w)&(u-w)&-(u-w)&\undermat{3m}{w&\cdots&w}
	\end{array}\right) .$$  \\
	\textbf{Case 3:} Let  $n=3m+5$ for some  integer $m\geq 0$. Then choose   $$\mathcal{G}=\left( \begin{array}{cccccccc}
		u&u&u&u&w&u&\cdots&u\\
		(u-w)&-(u-w)&(u+w)&-(u+w)&z&\undermat{3m}{w&\cdots&w}
	\end{array}\right) .$$  \\
		In each case, an additive code generated by $\mathcal{G}$ is a one-rank hull code with distance $n$. Since the matrix $\log_\xi(\mathcal{G}\odot_M\mathcal{G}^T)$  is 1 rank matrix of the type $\begin{pmatrix}
		a&0\\
		*&0
	\end{pmatrix}$, $\begin{pmatrix}
		0&0\\
		0&a
	\end{pmatrix}$ and $\begin{pmatrix}
		0&0\\
		*&a
	\end{pmatrix}$, respectively, where $a\neq 0$ and $*$ could be  0 or non-zero entry.
\end{proof}
\begin{remark}
	If $n=2$, then $d_1[2,2]_{9,M}\geq 1 $ for all dualities where  $\chi_u(u)\neq1$ for all $u\in\mathbb{F}^*_{9}$.
\end{remark}
\begin{theorem}\label{non-symmetric n-1}
	For non-symmetric dualities over $\mathbb{F}_4$, we have $d_1[n,2]_{4,M}=n-1$.
\end{theorem}
\begin{proof}
	Observe that for non-symmetric dualities over $\mathbb{F}_4$, we have $\chi_x(x)=-1$, for all $x\in \mathbb{F}^*_{4}$ and $\chi_x(y)=-\chi_y(x)$,  for all $x\neq y\in\mathbb{F}^*_4$. Let $\mathcal{G}=\begin{pmatrix}
		u_1&u_2&\cdots&u_n\\
		v_1&v_2&\cdots&v_n\\
	\end{pmatrix}$ be a generator matrix of a one-rank hull code $C$ with distance $n$. This implies $u_i\neq 0$, $v_i\neq 0$ and $u_i\neq v_i$, for all $1\leq i\leq n$. Then $$\mathcal{G}\odot_M\mathcal{G}^T=\begin{pmatrix}
	\prod_{i=1}^{n}\chi_{u_i}(u_i)&\prod_{i=1}^{n}\chi_{u_i}(v_i)\\
	\prod_{i=1}^{n}\chi_{v_i}(u_i)&\prod_{i=1}^{n}\chi_{v_i}(v_i)
	\end{pmatrix}=\begin{pmatrix}
(-1)^n&\prod_{i=1}^{n}\chi_{u_i}(v_i)\\
	(-1)^n\prod_{i=1}^{n}\chi_{u_i}(v_i)&(-1)^n
		\end{pmatrix}$$
		If $n$ is even, then $\log_\xi(\mathcal{G}\odot_M\mathcal{G}^T)$ is equal to either $\begin{pmatrix}
			0&0\\0&0
		\end{pmatrix}$ or $\begin{pmatrix}
		0&1\\1&0
		\end{pmatrix}$.  If $n$ is odd, then $\log_\xi(\mathcal{G}\odot_M\mathcal{G}^T)$ is equal to either $\begin{pmatrix}
		1&0\\1&1
		\end{pmatrix}$ or $\begin{pmatrix}
		1&1\\0&1
		\end{pmatrix}$. In all cases, $C$ is a not one-rank hull code, which is a contradiction. \\ Let $\mathcal{G}=\begin{pmatrix}
		x&x&\cdots&x&0\\
		y&y&\cdots&y&y\\
		\end{pmatrix}$ be a generator matrix of an additive code $C$ of length $n$. Then for any non-symmetric duality, if $n$ is odd, then $\log_\xi(\mathcal{G}\odot_M\mathcal{G}^T)$ is equal to $\begin{pmatrix}
		0&0\\0&1
		\end{pmatrix}$,  if $n$ is even, then $\log_\xi(\mathcal{G}\odot_M\mathcal{G}^T)$ is equal to either $\begin{pmatrix}
		1&1\\0&0
		\end{pmatrix}$ or $\begin{pmatrix}
		1&0\\1&0
	\end{pmatrix}$. In all cases, $C$ is a one-rank hull code with distance $n-1$. Hence, $d_1[n,2]_{4,M}=n-1$ for non-symmetric dualities.
\end{proof}
\begin{theorem}\label{symmetric n-1}
	For symmetric dualities over $\mathbb{F}_4$, we have $d_1[n,2]_{4,M}=n-1$.
\end{theorem}
\begin{proof}
	Observe that for symmetric dualities over $\mathbb{F}_4=\{0,u,v,u+v\}$, we have $\chi_u(u)=1, \ \chi_v(v)=-1,\ \chi_u(v)=-1$ (Since we are considering non skew-symmetric dualities). Let $\mathcal{G}=\begin{pmatrix}
		u_1&u_2&\cdots&u_n\\
		v_1&v_2&\cdots&v_n\\
	\end{pmatrix}$ be a generator matrix of a one-rank hull code $C$ with distance $n$. This implies $u_i\neq 0$, $v_i\neq 0$ and $u_i\neq v_i$, for all $1\leq i\leq n$.  Suppose in the first row  $u$ appears $m$ times. If $m$ is even, remove all $m$ columns  containing $u$ in the first row from $\mathcal{G}$, say it is $\mathcal{G}'$. Then we have
	\begin{align*}
	\mathcal{G}\odot_M\mathcal{G}^T&=\begin{pmatrix}
		\prod_{u_i\neq u}\chi_{u_i}(u_i)(\chi_u(u))^m&\prod_{u_i\neq u}\chi_{u_i}(v_i)\prod\chi_u(v_j)\\
		\prod_{u_i\neq u}\chi_{v_i}(u_i)\prod\chi_{v_j}(u)&\prod_{u_i\neq u}\chi_{v_i}(v_i)\prod_{u_i= u}\chi_{v_j}(v_j)
	\end{pmatrix}\\\\
	&=\begin{pmatrix}
		\prod_{u_i\neq u}\chi_{u_i}(u_i)(1)^m&\prod_{u_i\neq u}\chi_{u_i}(v_i)(-1)^m\\
		\prod_{u_i\neq u}\chi_{v_i}(u_i)(-1)^m&\prod_{u_i\neq u}\chi_{v_i}(v_i)(-1)^m
	\end{pmatrix}\\\\
	&= \begin{pmatrix}
		\prod_{u_i\neq u}\chi_{u_i}(u_i)&\prod_{u_i\neq u}\chi_{u_i}(v_i)\\
		\prod_{u_i\neq u}\chi_{v_i}(u_i)&\prod_{u_i\neq u}\chi_{v_i}(v_i)
	\end{pmatrix}= \mathcal{G}'\odot_M\mathcal{G}'^T .
	\end{align*} 
	Hence, $\mathcal{G}'$ is a generator matrix of  one-rank hull code with distance $n-m$. If $m$ is odd,  remove $m-1$ (even) columns  containing $u$. Then $\mathcal{G}'$  is a generator matrix of one-rank hull code with distance $n-(m-1)$. Similarly, we do for the second row of $\mathcal{G}$. Therefore, we have at most one $u$ in both the rows of $\mathcal{G}'$.
	\begin{enumerate}
		\item There is no $u$ in both rows of $\mathcal{G}'$ then $\log_\xi(\mathcal{G}'\odot_M\mathcal{G}'^T)$ is equal to either $\begin{pmatrix}
			1 &0\\0& 1
		\end{pmatrix}$ or $\begin{pmatrix}
		0 &0\\0& 0
		\end{pmatrix}$, depending on $n$ is even or odd.
		
			\item There is one $u$ in both rows of $\mathcal{G}'$ then $\log_\xi(\mathcal{G}'\odot_M\mathcal{G}'^T)$ is equal to either $\begin{pmatrix}
			1 &0\\0& 1
		\end{pmatrix}$ or $\begin{pmatrix}
			0 &0\\0& 0
		\end{pmatrix}$, depending on $n$ is even or odd.
		
			\item There is one $u$ either in first or second row of $\mathcal{G}'$ then $\log_\xi(\mathcal{G}'\odot_M\mathcal{G}'^T)$ is equal to either $\begin{pmatrix}
			1 &1\\1& 0
		\end{pmatrix}$ or $\begin{pmatrix}
			0 &1\\1& 1
		\end{pmatrix}$, depending on $n$ is even or odd.
	\end{enumerate}
	In all the above cases, $\mathcal{G}'$ can not generate a one-rank hull code, which is a contradiction. Therefore, 	$d_1[n,2]_{4,M}\leq n-1$ for symmetric dualities.\\
	If $n$ is even then $\mathcal{G}=\begin{pmatrix}
		0&u&\cdots&u&u\\
		v&v&\cdots&v&0
	\end{pmatrix}$, and if $n$ is odd then  $\mathcal{G}=\begin{pmatrix}
	0&u&\cdots&u&u\\
	v&v&\cdots&v&v
	\end{pmatrix}$  generate an additive code $C$.  It is clear that $C$ is a one-rank hull code of length $n$ and distance $n-1$. Hence the result follows.
\end{proof}
\begin{theorem}
	For all dualities $M$ over $\mathbb{F}_{p^2}$, $p\neq 2,3$, we have $d_1[n,2n-2]_{p^2,M}=2$.
\end{theorem}
\begin{proof}
	Let $\mathcal{G}=\begin{pmatrix}
		x_1&x_2&\cdots&x_n\\
			y_1&y_2&\cdots&y_n
	\end{pmatrix}$ be a generator matrix of one-rank hull code $C^{M^T}$ with distance $n$ then the sets $\{x_i,y_i\}$ are 2-linear independent for all $i$. And, $$C=\{(u_1,u_2,\dots,u_n)\in\mathbb{F}^n_{p^2}|\prod_{i=1}^{n}\chi_{u_i}(x_i)=1\text{\ and\ }\prod_{i=1}^{n}\chi_{u_i}(y_i)=1\}.$$ If, $\textbf{z}=(0,\cdots,0,z_i,0,\cdots,0)\in C$ is a vector of weight one, then $\chi_{z_i}(x_i)=1$ and $\chi_{z_i}(y_i)=1$. Which is not possible over $\mathbb{F}_{p^2}$ since $\{x_i,y_i\}$ is 2-linear independent set. Hence,  $d_1[n,2n-2]_{p^2,M}\geq2$. From the Singleton bound for the additive code $[n, p^{(2n-2)}]$, we have $d \leq 2$. Hence, combining the results, we get $d_1[n,2n-2]_{p^2,M}=2$.
\end{proof}
\begin{theorem}\label{2n-2}
	For all dualities $M$ over $\mathbb{F}_4$, we have $d_1[n,2n-2]_{4,M}=1$.
\end{theorem}
\begin{proof}
	Let $C^{M^T}$ be an $[n,2^2]$ one-rank hull code with generator matrix $\mathcal{G}=\begin{pmatrix}
		x_1&x_2&\cdots&x_n\\
		y_1&y_2&\cdots&y_n
	\end{pmatrix}$. Then, $$C=\{(u_1,u_2,\dots,u_n)\in\mathbb{F}^n_{4}|\prod_{i=1}^{n}\chi_{u_i}(x_i)=1\text{\ and\ }\prod_{i=1}^{n}\chi_{u_i}(y_i)=1\}.$$
	 Furthermore, according to  Theorem \ref{non-symmetric n-1} and \ref{symmetric n-1}, we have  $d(C^{M^T})\leq n-1$. This implies that  for some $i$, either $x_i=0$ or $y_i=0$ or $x_i=y_i$. In each case, by Theorem \ref{chi u(v)=1}, we have a vector $(0,\cdots,0,z_i,0,\cdots,0)\in C$. Hence, $d_1[n,2n-2]_{4,M}=1$.
\end{proof}
During our search in MAGMA \cite{Magma}, we identified improved
parameters for quaternary one-rank hull codes over non-symmetric dualities. This indicates the non-symmetric dualities over $\mathbb{F}_4$ hold significant interest. Here, we provide  the highest possible  minimum distances for quaternary one-rank hull codes under non-symmetric dualities (see Table \ref{table 1}). Over $\mathbb{F}_4$, there are two non-symmetric dualities, $N_1$ and $N_2$, such that $N_2^T=N_1$.

\begin{center}
\begin{tabular}{|c|c|c|c|c|}
	\hline
	$ N_1 $& 0 &1  &$ \omega $  &$ \omega+1 $  \\
	\hline
	0&1  & 1 & 1 & 1 \\
	\hline
	1& 1 & -1 & -1 & 1 \\
	\hline
	$ \omega$&  1 & 1 &-1  &  -1 \\
	\hline
	$ \omega+1 $& 1 &  -1&  1& -1 \\
	\hline
\end{tabular} \hspace{0.5cm}
\begin{tabular}{|c|c|c|c|c|}
	\hline
	$ N_2 $& 0 &1  &$ \omega  $ &$ \omega+1 $  \\
	\hline
	0& 1 & 1 &  1& 1 \\
	\hline
	1&1  & -1 & 1 & -1 \\
	\hline
	$ \omega $& 1 & -1&-1  & 1 \\
	\hline
	$\omega+1 $& 1 & 1 & -1 & -1 \\
	\hline
\end{tabular}
\end{center}

According to Lemma \ref{C is one-rank wrt M^T}, if $C$ is a one-rank hull code under $N_1$, then so does $N_2$. Hence, for both duality, we have the same table. Although MAGMA includes  built-in functions for determining the dual space of additive codes, it lacks a built-in function for the dualities that we are using. As a result, we have developed an algorithm over $\mathbb{F}_4$ to determine if a given additive code is a one-rank hull code. Our search was carried out in the MAGMA software package \cite{Magma}. Generator matrices for one-rank hull codes for a non-symmetric  duality can be found online at \href{https://drive.google.com/file/d/132KW2UrlpdSEfr2lUSfqswEAwvSU0QMq/view?usp=sharing}{https://drive.google.com/file/d/132KW2UrlpdSEfr2lUSfqswEAwvSU0QMq/view?usp=sharing}.\\

\vspace{0.1cm}
 The one-rank hull codes presented here are either optimal or near to optimal according to \cite{small_additive}. An $[n,p^k,d]$ code is said to be an optimal one-rank hull code if $d=d_1[n,k]$ for a given $n$ and $k$. From Theorems \ref{F4_1}, \ref{ne-1}, \ref{non-symmetric n-1}, and \ref{2n-2}, we have   $d_{1}[n, 1]=\begin{cases}
 	n & \ \text{if $n$ is even}\\
 	n-1 &\  \text{if $n$ is odd}
 \end{cases}$, $d_1[n,2]=n-1$, $d_1[n,2n-1]=1$, and $d_1[n,2n-2]=1$ for non-symmetric dualities over $\mathbb{F}_4$.  In the Table \ref{table 1}, optimal codes are highlighted in bold. Also, we identify some optimal quaternary one-rank hull codes under non-symmetric dualities which were not optimal under all symmetric dualities, see Table 2,3,4,5 in \cite{onerank}. For example,  $[4,2^4,3]$, $[5,2^4,4]$, $[8,2^4,6]$, $[9,2^4,7]$, $[10,2^4,8]$, $[10,2^5,7]$, $[6,2^6,3]$, $[8,2^6,5]$, $[8,2^8,4]$, $[10,2^9,5]$ and $[10,2^{14},3]$. In the table, an asterisk ($^*$) represents these codes.
\begin{table}[h]                                                                                                                                                                                                                                                                                                                                                                                                                                                                                                                                       In the table, an asterisk (*) represents these codes.
\caption{Highest minimum distance  for quaternary one-rank hull codes with respect to non-symmetric dualities. The $*$ signifies the improvement over the prior works. }
		\begin{tabular}{c||ccccccccccccccccccc}
			\hline
			n/k&1&2&3&4&5&6&7&8&9&10&11&12&13&14&15&16&17&18&19\\
		\hline
			1&-&-\\
			
			2&\textbf{2}&\textbf{1}&\textbf{1}&-\\
		
			3&\textbf{2}&\textbf{2}&\textbf{2}&\textbf{1}&\textbf{1}&-\\
			
			4&\textbf{4}&\textbf{3}&\textbf{3}&\textbf{3$^*$}&\textbf{2}&\textbf{1}&\textbf{1}&-\\
		
			5&\textbf{4}&\textbf{4}&\textbf{4}&\textbf{4$^*$}&3&\textbf{3$^*$}&\textbf{2}&\textbf{1}&\textbf{1}&-\\
			
			6&\textbf{6}&\textbf{5}&\textbf{5}&\textbf{4}&\textbf{4}&3&\textbf{3}&\textbf{2}&\textbf{2}&\textbf{1}&\textbf{1}&-\\
		
			7&\textbf{6}&\textbf{6}&\textbf{6}&\textbf{5}&4&\textbf{4}&3&\textbf{3}&\textbf{3}&\textbf{2}&\textbf{2}&\textbf{1}&\textbf{1}&-\\
			
			8&\textbf{8}&\textbf{7}&\textbf{6}&\textbf{6$^*$}&5&\textbf{5$^*$}&\textbf{4}&\textbf{4$^*$}&\textbf{3}&\textbf{3}&2&\textbf{2}&\textbf{2}&\textbf{1}&\textbf{1}&-\\
		
			9&\textbf{8}&\textbf{8}&\textbf{7}&\textbf{7$^*$}&\textbf{6}&5&\textbf{5}&4&\textbf{4}&3&\textbf{3}&\textbf{3}&2&\textbf{2}&\textbf{2}&\textbf{1}&\textbf{1}&-\\
		
			10&\textbf{10}&\textbf{9}&\textbf{8}&\textbf{8$^*$}&\textbf{7$^*$}&\textbf{6}&5&5&\textbf{5$^*$}&4&\textbf{4}&3&\textbf{3}&\textbf{3$^*$}&\textbf{2}&\textbf{2}&\textbf{2}&\textbf{1}&\textbf{1}\\
			\hline
	\end{tabular}
	\label{table 1}
	\end{table}
\section{Conclusion}
In this article, we have discussed additive  one-rank hull codes with respect to arbitrary dualities over finite fields. The main contributions of this article are the following:
\begin{enumerate}
	\item A novel approach has been given to find the one-rank hull codes by establishing a  connection between self-orthogonal elements and solutions of quadratic forms.
	\item  Precise count of self-orthogonal elements with respect to arbitrary dualities over finite fields of odd characteristics has been provided.
	\item Construction methods have been given for small rank hull codes.
	\item The value of  the highest possible minimum distances $d_1[n,k]_{p^e,M}$ for $k=1,2$ and $n\geq 2$ has been determined.
	\item Optimal quaternary one-rank hull codes $[4,2^4,3]$, $[5,2^4,4]$, $[8,2^4,6]$, $[9,2^4,7]$, $[10,2^4,8]$, $[10,2^5,7]$, $[6,2^6,3]$, $[8,2^6,5]$, $[8,2^8,4]$, $[10,2^9,5]$, and $[10,2^{14},3]$ over non-symmetric dualities are identified (see Table \ref{table 1}), which improve the minimal distance of the quaternary one-rank hull codes over all symmetric dualities see Table 2,3,4,5 in \cite{onerank}.
	
\end{enumerate}  
  As a future  work, a precise count of self-orthogonal elements over fields of even characteristics can be determined. It is known that determining the minimum distance of the codes is very important to estimate the error correcting capability.  Hence, the valuable problem could be, to find out a good upper bound and lower bound for $d_1[n,k]_{p^e,M}$ over arbitrary dualities. 
  \section*{Acknowledgments}
  The first author is supported by UGC, New Delhi, Govt. of India under grant DEC18-417932.
  The second author is ConsenSys Blockchain chair professor. He thanks ConsenSys AG for that  privilege.
  \section*{Statements and Declarations}
  \textbf{Author’s Contribution}: All authors contributed equally to this work.\\
  \textbf{Conflict of interest:} The authors have no conflicts of interest in this paper.
 
	\bibliographystyle{abbrv}
\bibliography{references}

	\end{document}